\documentclass[conference,a4paper]{IEEEtran}

\usepackage[utf8]{inputenc} 
\usepackage[T1]{fontenc}
\usepackage{url}              
\usepackage{cite}             

\usepackage[cmex10]{amsmath}  
\usepackage{mleftright}       
\mleftright                   

\usepackage{graphicx}         
\usepackage{booktabs}         

\usepackage{subcaption}

\usepackage{caption}

\usepackage{sidecap}

\usepackage{hyperref}
\hypersetup{colorlinks,allcolors=black}
\usepackage{scalerel}
\usepackage[section]{placeins}
\usepackage[utf8]{inputenc} 
\usepackage[T1]{fontenc}
\usepackage{url}
\usepackage{ifthen}
\usepackage{cite}
\usepackage[cmex10]{amsmath} 
\usepackage{amssymb,lipsum} 
\usepackage{ bm}
\allowdisplaybreaks
\usepackage{amsthm}                       
\usepackage{tikz}
\usetikzlibrary{patterns}
\usepackage{pgfplots}
\pgfplotsset{compat=1.17}

\pgfplotsset{
    legend image with text/.style={
        legend image code/.code={%
            \node[anchor=center] at (0cm,0cm) {#1};
        }
    },
}
\usepgfplotslibrary{fillbetween}
\usepackage{graphicx}
\usepackage{url}
\usepackage{widetext}
\usepackage{float, stfloats}
\usepackage{latexsym}

\newcommand{\cS}{\bm{\mathcal{S}}}

\newcommand{\bX}{\bm{X}}
\newcommand{\bS}{\bm{S}}
\newcommand{\bSig}{\bm{\Sigma}}
\newcommand{\bGam}{\bm{\Gamma}}
\newcommand{\bY}{\bm{Y}}
\newcommand{\bZ}{\bm{Z}}
\newcommand{\bA}{\bm{A}}
\newcommand{\bO}{\bm{O}}
\newcommand{\bU}{\bm{U}}
\newcommand{\bV}{\bm{V}}

\newcommand{\bG}{\bm{G}}
\newcommand{\bB}{\bm{B}}

\newcommand{\cY}{\bm{\mathcal{Y}}}

\newcommand{\cZ}{\bm{\mathcal{Z}}}

\newcommand{\bu}{\bm{u}}
\newcommand{\bw}{\bm{w}}
\newcommand{\ba}{\bm{a}}

\newcommand{\br}{\bm{r}}
\newcommand{\bl}{\bm{l}}
\newcommand{\bb}{\bm{b}}
\newcommand{\bbeta}{\bm{\eta}}
\newcommand{\bE}{\mathbb{E}}
\newcommand{\ci}{\mathsf{i}}

\newcommand{\bs}{\mathbf{s}}
\newcommand{\bM}{\bm{M}}
\newcommand{\bv}{\bm{v}}

\newcommand{\bR}{\mathbb{R}}

\newcommand{\sH}{\mathsf{H}}
\newcommand{\bI}{\bm{I}}
\DeclareMathOperator{\Tr}{Tr}

\newcommand{\bsig}{\bm{\sigma}}

\newtheorem{theorem}{Theorem}

\newtheorem{lemma}{\textbf{Lemma}}
\newtheorem{assumption}{Assumption}

\newtheorem{statement}[theorem]{\textbf{Statement}}
\newtheorem{proposition}{Proposition}
\begin{document}

\title{Rectangular Rotational Invariant Estimator for General Additive Noise Matrices} 

\author{\IEEEauthorblockN{Farzad Pourkamali and Nicolas Macris}
\IEEEauthorblockA{SMILS, EPFL, Lausanne, Switzerland}
Emails: \{farzad.pourkamali,nicolas.macris\}@epfl.ch}

\maketitle

\begin{abstract}
 We propose a \textit{rectangular rotational invariant estimator} to recover a real matrix from noisy matrix observations coming from an arbitrary additive rotational invariant perturbation, in the large dimension limit. Using the Bayes-optimality of this estimator, we derive the asymptotic minimum mean squared error (MMSE). For the particular case of Gaussian noise, we find an explicit expression for the MMSE in terms of the limiting singular value distribution of the observation matrix. Moreover, we prove a formula linking the asymptotic mutual information and the limit of log-spherical integral of rectangular matrices. We also provide numerical checks for our results, which match our theoretical predictions and known Bayesian inference results.
\end{abstract}

\section{Introduction}
Matrix denoising is the problem of removing noise from a given data matrix while preserving important features or structure of the signal. 
We consider a simple, yet general setting, where the noise is additive but can be {\it structured}. Suppose the ground-truth matrix $\bS \in \bR^{N \times M}$ is distributed according to a rotationally invariant prior, i.e. $P_{S}(\bS) = P_{S}( \bU \bS \bV^T)$ for any orthogonal matrices $\bU \in \bR^{N \times N}, \bV \in \bR^{M \times M}$. The matrix $\bS$ is corrupted by an additive noise, and we observe:
\begin{equation}
    \bY = \sqrt{\lambda} \bS + \bZ 
    \label{observation-matrix}
\end{equation}
where $\bZ \in \bR^{N \times M}$ is distributed according to a rotationally invariant prior (not necessarily Gaussian), and $\lambda \in \bR_+$ is proportional to the signal-to-noise-ratio (SNR). We assume that $M$ scales like $N$, and $N/M \to \alpha$. Moreover, we assume that the entries of $\bS$ and $\bZ$ are of the order $O(\frac{1}{\sqrt{N}})$. This scaling is such that the singular values of $\bS$, $\bZ$ and $\bY$ are of the order $O(1)$ as $N \to \infty$. 
Studying the problem for the case $\alpha \in (0,1]$ suffices. Indeed, suppose the observation matrix $\bY \in \bR^{N \times M}$ has dimensions $N>M$ (so $\alpha >1$),  then exchanging the role of $M, N$, we can apply our results to the matrix $\bY^T$ with aspect ratio $1/\alpha \in (0,1)$. 

Given its fundamental role, model \eqref{observation-matrix} has gained much attention from the theoretical and algorithmic point of view. For a noise matrix $\bZ$ with i.i.d. Gaussian entries, when the hidden signal $\bS$ is low-rank ($O(1)$ compared to the size), the problem is well studied, and its fundamental limits have been derived under various settings, see \cite{korada2009exact, dia2016mutual, lelarge2019fundamental, pourkamali2021mismatched, barbier2022price, lesieur2015mmse, barbier2019adaptive, barbier2019adaptive-b, miolane2017fundamental, luneau2020high, pourkamali2022mismatched}. Recently, the rotationally invariant noise and low-rank signal is studied in \cite{fan2022approximate, barbier2022bayes}, and an optimal AMP-based algorithm is proposed in \cite{barbier2022bayes}.

In the setting where the rank of the signal $\bS$ grows with the dimension $N$ and the noise is Gaussian, the problem has been studied in \cite{kabashima2016phase, barbier2022statistical, maillard2022perturbative, troiani2022optimal}, and denoising algorithms are proposed in \cite{troiani2022optimal, camilli2022matrix, bodin2023gradient}. In the symmetric seting, where both $\bS$ and $\bZ$ are symmetric (and hence $N=M$), and general rotationally invariant noise, the Rotationally Invariant Estimator (RIE) has been proposed in \cite{bun2016rotational} which is be shown to be Bayes-optimal. 

In this contribution, we propose a \textit{rectangular RIE}, which is the generalization of the RIE introduced in \cite{bun2016rotational} to rectangular additive noise matrices.  An estimator $\hat{\Xi}(\bY)$, is called a Rotational Invariant Estimator (RIE) if for any orthogonal matrices $\bU \in \bR^{N \times N}, \bV \in \bR^{M \times M}$, we have:
\begin{equation*}
    \bU \hat{\Xi}(\bY) \bV^T = \hat{\Xi}(\bU \bY \bV^T)
\end{equation*}
In this case, it turns out that the singular vectors of $\hat{\Xi}(\bY)$ are the same as those of the matrix $\bY$ \cite{stein1975estimation,takemura1984orthogonally}. Consequently, as discussed in section \ref{RIE}, the best RIE depends on the overlap between the singular vectors of $\bS$,$\bY$. This overlap can be computed in the large $N$ limit using results from random matrix theory and the replica trick, and eventually we obtain an RIE which only depends on the observation matrix and the knowledge of the noise distribution. Authors in \cite{bun2016rotational} call that type of estimator "miracle" estimator because it does not require any knowledge of the signal distribution and at the same time it is Bayesian-optimal. Moreover, using Bayesian-optimality of this estimator, we can compute the MMSE of estimation in model \eqref{observation-matrix}.

For the particular case of Gaussian noise, we are able to compute an explicit expression for the MMSE. This enables us to access the mutual information via the I-MMSE relation \cite{guo2005mutual} valid for Gaussian noise. Alternatively, under suitable assumptions, we can {\it prove} by independent methods that the mutual information is linked to the asymptotic log-spherical integral, which has been studied in physics and mathematics literature \cite{guionnet2002large, guionnet2021large}. Therefore, with our analysis, we are able derive the asymptotic log-spherical integral in a special case.

 \textit{Organization}: We begin section \ref{result} with a brief introduction to the random matrix theory tools we use, then we state our main results. In section \ref{RIE}, we sketch the derivation of the rectangular RIE. Section \ref{MMSE-computation} describes the computation of the MMSE, followed by numerical results in section \ref{numericals}.

\textit{Notation}: For the vector $\bsig \in \bR^N$ (and $\bm{\gamma}, \bm{\xi}$), $\bSig \in \bR^{N \times M}$ (and $\bm{\Gamma}, \bm{\Xi}$) denotes a matrix constructed as $
    \bSig = \left[
\begin{array}{c|c}
\bSig_N & \mathbf{0}_{N \times (M-N)}
\end{array}
\right]
$ with $\bSig_N \in \bR^{N \times N}$ a diagonal matrix with diagonal $\bsig$. For a sequence of matrices $\bA$ of growing size, we denote the limiting empirical singular value distribution (ESD) by $\mu_A$, and the limiting eigenvalue distribution $\bA \bA^T$ by $\rho_A$.
The free rectangular convolution \cite{benaych2009rectangular} with ratio $\alpha \in [0,1]$  of two probability distributions is denoted 
 by $\boxplus_\alpha$.

\section{Main Results}\label{result}
\subsection{Preliminaries on Random Matrix Theory}
The ESD of $\bS$ is defined as:
\begin{equation*}
    \mu_{S}^{(N)}(x) = \frac{1}{N} \sum_{i=1}^N \delta(x - \sigma_i^S )
\end{equation*}
where $(\sigma_i^S)_{1 \leq i \leq N}$ are the singular values of $\bS$.

For a probability measure $\mu$ with support contained in $[-K, K ]$ with $K>0$, we define a generating function of (even) moments $\mathcal{M}_{\mu} : [0,K^{-2}] \to \bR_+$
as
\begin{equation*}
    \mathcal{M}_{\mu} (z) = \int \frac{1}{1-t^2 z} \mu(t) \, d t - 1
\end{equation*}
For $\alpha \in [0,1]$, define $T^{(\alpha)}(z) = (\alpha z +1)(z+1)$, and $\mathcal{H}_{\mu}^{(\alpha)}(z) = z T^{(\alpha)}\big(\mathcal{M}_{\mu} (z)\big)$. The \textit{rectangular R-transform} is defined as:
\begin{equation*}
    \mathcal{C}_{\mu}^{(\alpha)}(z) = {T^{(\alpha)}}^{-1}\Big(\frac{z}{{\mathcal{H}_{\mu}^{(\alpha)}}^{-1}(z)} \Big)
\end{equation*}

For a probability density $\mu(x)$ on $\bR$, the \textit{Stieltjes}  (or \textit{Cauchy}) transform  is defined as 
\begin{equation*}
   \mathcal{G}_\mu (z) = \int_{\bR} \frac{1}{z - x} \mu(x) \, dx \hspace{10 pt} \text{for } z \in \mathbb{C} \backslash {\rm supp}(\mu)
\end{equation*}
By  Plemelj formulae we have for $x\in \mathbb{R}$,
\begin{equation}
    \lim_{y \to 0^+} \mathcal{G}_\mu(x - \ci y) = \pi \sH [\mu](x) + \pi \ci \mu(x) 
    \label{Plemelj formulae}
\end{equation}
with $\sH [\mu](x) = {\rm p.v.} \frac{1}{\pi} \int_{\bR} \frac{\mu(t)}{x - t}  d t$ the \textit{Hilbert} transform of $\mu$.
\subsection{Rectangular RIE}
Given the matrix $\bY \in \bR^{N \times M}$ with svd $\bY = \bU_Y \bm{\Gamma} \bV_Y^T$, the \textit{rectangular RIE} is constructed as:
    \begin{equation*}
        \Xi^*(\bY) = \bU_Y \bm{\Xi}^* \bV_Y^T
    \end{equation*}
with $\bm{\Xi}^* $ is a rectangular diagonal matrix of singular values $\xi_i^*$ ($1 \leq i \leq N$):
\begin{equation}
\begin{split}
       \xi_i^* &=\frac{1}{\sqrt{\lambda}} \Bigg[ \gamma_i - \frac{1}{\pi \bar{\mu}_{Y}(\gamma_i)} {\rm Im} \, \mathcal{C}^{(\alpha)}_{\mu_Z}\bigg( \frac{1- \alpha}{\gamma_i} \pi \sH [\bar{\mu}_{Y}](\gamma_i) \\
       &\hspace{10pt}+ \alpha \big( \pi \sH [\bar{\mu}_{Y}](\gamma_i)\big)^2  - \alpha \big( \pi \bar{\mu}_{Y}(\gamma_i)\big)^2 \\ 
       &\hspace{20pt}+\ci \pi \bar{\mu}_{Y}(\gamma_i) \big(\frac{1-\alpha}{\gamma_i} + 2 \alpha \pi \sH [\bar{\mu}_{Y}](\gamma_i) \big) \bigg) \Bigg]
\end{split}
\label{rect-RIE}
\end{equation}
where $ \bar{\mu}_{Y}(\gamma) = \frac{1}{2}(\mu_Y(\gamma) + \mu_Y(-\gamma))$ is the symmetrization of the limiting ESD of $\bY$.
Note that 
this estimator does not require any information about the prior $P_{S}(\bS)$.

The Mean Squared Error (MSE) is defined as usual
\begin{equation*}
    {\rm MSE}_N(\hat{\Xi}) = \frac{1}{N} \bE \Big[ \big\| \bS - \hat{\Xi}(\bY) \big\|_F^2 \Big]
\end{equation*}
where the expectation is over the prior $P_S$ and noise distribution. The Minimum MSE (${\rm MMSE}_N$) is the best possible reconstruction error, and is achieved by $\bE[\bS | \bY ]$.

One can see that for  model \eqref{observation-matrix}, the posterior mean estimator is a RIE, therefore since the estimator \eqref{rect-RIE} is (conjectured to be) the best among RIE class, we can conclude that it is optimal, i.e. its MSE is equal to the MMSE. The rhs in \eqref{rect-RIE} is a {\it function} of singular values of $\bY$ and is denoted by $\xi^* : {\rm supp}(\mu_Y)\to \mathbb{R}$.

\begin{statement}[\textbf{MMSE}]\label{General-MMSE-statement}
    Assume that the ESD of $\bS$, $\bZ$ converge to well-defined measures $\mu_S, \mu_Z$ with bounded second moments. We have:
    \begin{equation}
        \lim_{N \to \infty} {\rm MMSE}_N(\lambda) = \int x^2 \mu_S(x) \, dx - \int {\xi^*(x)}^2 \mu_Y(x) \, dx
        \label{MMSE-eq}
    \end{equation}
    where $\mu_Y$ is the limiting ESD of $\bY$, $\mu_Y = \mu_S \boxplus_{\alpha} \mu_Z$.
\end{statement}

\subsection{Gaussian Noise}
Consider the noise matrix $\bZ$ with i.i.d. Gaussian entries of variance $\frac{1}{N}$. In this case, $\mathcal{C}^{(\alpha)}_{\mu_Z}(z) = \frac{1}{\alpha}z$, and the estimator in \eqref{rect-RIE} reduces to:
\begin{equation}
    \xi_i^* = \frac{1}{\sqrt{\lambda}} \Big[ \gamma_i - \big( \frac{1-\alpha}{\alpha} \frac{1}{\gamma_i} + 2 \pi \sH [\bar{\mu}_{Y}](\gamma_i) \big) \Big]
    \label{G-RIE}
\end{equation}
This estimator was previously derived in \cite{troiani2022optimal}  using the Feynman-Hellman theorem.
Given this rather simple expression for the optimal singular values, we can compute the asymptotic MMSE for the particular case of Gaussian noise.

\begin{statement}[\textbf{Gaussian MMSE}]\label{Gaus-MMSE-statement}
    Assume that the ESD of $\bS$ converges to a well-defined measure $\mu_S$ with compact support and bounded second moment. Under Gaussian noise we have:
    \begin{equation}
    \begin{split}
        \lim_{N \to \infty} {\rm MMSE}_N(\lambda) = \frac{1}{\lambda} \Big[ \frac{1}{\alpha} - &\big(\frac{1}{\alpha} -1 \big)^2 \int \frac{\mu_Y(x)}{x^2}\, dx \\
        &- \frac{\pi^2}{3} \int {\mu_Y(x)}^3 \, dx \Big]
    \end{split}
        \label{G-MMSE-eq}
    \end{equation}
    where $\mu_Y$ is the limiting ESD of $\bY$, $\mu_Y = \mu_S \boxplus_{\alpha} \mu_{\rm MP}$.
\end{statement}

From the asymptotic MMSE, we can access the asymptotic mutual information $ \frac{1}{M N} \mathcal{I}_N(\bS; \bY)$, using the I-MMSE relation \cite{guo2005mutual} which for model \eqref{observation-matrix} states:
\begin{equation*}
     {\rm MMSE}_N(\lambda) = 2 \frac{M}{N} \frac{d}{d \lambda} \frac{1}{M N} \mathcal{I}_N(\bS; \bY)
\end{equation*}
 Concavity of the mutual information w.r.t. the SNR, implies this relation also holds in the limit $N \to \infty$ (assuming limits exist).  Therefore, it suffices to compute the integral of the asymptotic MMSE over $\lambda$ to find the asymptotic mutual information. On the other hand, by an independent analysis, we {\it prove} that the asymptotic mutual information is linked to an asymptotic spherical integral. The rectangular spherical integral is defined for two matrices $\bA, \bB \in \bR^{N \times M}$ as:
 \begin{equation*}
    I_{N,M}(\bA, \bB) := \iint D \bU  \,  D \bV \, e^{N \Tr [\bA^T \bU  \bB \bV^T ]}
 \end{equation*}
where $D \bU, D \bV$ denotes the Haar measure over the groups of $N \times N$, $M \times M$ orthogonal matrices. The asymptotic behavior of these integrals has been studied in \cite{guionnet2021large} which proves that the limit $\lim_{N \to \infty} \frac{1}{N^2} \ln I_{N,M} (\bA, \bB)$ exists and equals a variational formula given in terms of limiting ESD of $\bA, \bB$. We use this result to prove the following theorem. Let $J[\mu_{\sqrt{\lambda} S}, \mu_{\sqrt{\lambda} S}\boxplus_{\alpha} \mu_{\rm MP}] = \lim_{N\to +\infty} \frac{1}{N^2} \ln I_{N,M} ( \sqrt{\lambda} \bS, \bY)$, where $\mu_{\sqrt{\lambda} S}$ is the limiting spectral distribution of $\sqrt{\lambda} \bS$, and $\mu_{\rm MP}$ is the  Marchenko-Pastur distribution. 
\begin{theorem}[\textbf{Mutual Information}]\label{theoremMain}
Under suitable assumptions on ESD of $\bS$, we have:
\vspace{-1mm}
\begin{equation}
\begin{split}
    \lim_{N \to \infty} \! \frac{1}{M N} \mathcal{I}_N(\bS; \bY)  \! &= \\
    & \hspace{-40pt} \lambda \alpha \int \! x^2 \mu_S(x) \, dx  - J[\mu_{\sqrt{\lambda} S}, \mu_{\sqrt{\lambda} S}\boxplus_{\alpha} \mu_{\rm MP}]
\end{split}
    \label{asymp-mI-th}
\end{equation}
\label{main-th}
\end{theorem}
\vspace{-2mm}
The assumptions in Theorem \ref{main-th} are mainly the assumption needed for the existence of the limit of the spherical integral stated in Theroem 1.1 in \cite{guionnet2021large}. Since the main focus of the current contribution is the RIE introduced in \eqref{rect-RIE}, we postpone the detailed proof of this theorem to a longer version paper (a sketch of the main steps is found in Appendix \ref{Proof-Thm1}).

\section{Rectangular RIE}\label{RIE}
To derive the estimator \eqref{rect-RIE}, we use the method of \cite{bun2016rotational}. 
Assuming a fixed ground-truth signal $\bS$ the model is equivalent to
\begin{equation}
   \bY = \bS + \bU \bZ \bV^T
   \label{RIE-observ-model}
\end{equation}
with $\bZ$ a fixed matrix with limiting singular value distribution $\mu_Z$, and $\bU \in \bR^{N \times N}, \bV \in \bR^{M \times M}$ random orthogonal matrices.
Given the data matrix $\bY = \bU_Y \bm{\Gamma} \bV_Y^T$, a RIE is constructed as $\hat{\Xi}(\bY) = \bU_Y \bm{\Xi} \bV_Y^T$. Note that for convenience the SNR parameter has been absorbed into $\bS$, so to obtain the estimator for model \eqref{observation-matrix}, it should be divided by $\sqrt{\lambda}$. Our goal is to have the minimum squared error w.r.t. the signal $\bS$. Let the singular values of $\bS$ be $\sigma_1, \dots, \sigma_N$. The squared error for this estimator can be written as
\begin{align}
    &\frac{1}{N} \big\| \bS - \Xi(\bY) \big\|_F^2 \label{MSE-expansion}
    \\ &= \frac{1}{N} \sum_{i=1}^N {\sigma_i}^2 + \frac{1}{N} \sum_{i=1}^N \xi_i^2 - \frac{2}{N} \sum_{i,j=1}^N \xi_i  \sigma_j \big(  \bu_i^T \bs_j^{(l)}\big) \big(  \bv_i^T \bs_j^{(r)} \big) \notag
\end{align}
where $\bs_i^{(l)}, \bs_i^{(r)}$ are the left and right singular vectors of $\bS$, and $\bu_i,\bv_i$ are columns of $\bU_Y$ and $\bV_Y$. Minimizing over $\xi_i$'s, we find :
\begin{equation}
    \xi_i^* = \sum_{j=1}^N \sigma_j \big(  \bu_i^T \bs_j^{(l)}\big) \big(  \bv_i^T \bs_j^{(r)} \big)
\label{optimal-sv}
\end{equation}
In \cite{bun2016rotational}, \eqref{optimal-sv} is called "oracle" estimator, since it requires the knowledge of the signal. We assume that in the large-$N$ limit, $\xi_i^*$'s can be approximated by the expectation, $\xi_i^* \approx \sum_{j=1}^N \sigma_j \Big \langle \big(  \bu_i^T \bs_j^{(l)}\big) \big(  \bv_i^T \bs_j^{(r)} \big) \Big \rangle$, where the expectation $\langle -\rangle$  is over the singular vectors of $\bY$. When the signal $\bS$ has finite-rank, the overlap has been computed in \cite{benaych2012singular}. For the large-rank, we discuss the derivation of the overlap in the following.

\subsection{Relation Between Overlap and the Resolvent}
\begin{figure*}[b!]
\centering
\begin{minipage}{\textwidth}
\medskip
\vspace{-5mm}
\hrule
\begin{equation}
    \cY = \left[
\begin{array}{ccc}
\hat{\bU}_Y & \hat{\bU}_Y &  \mathbf{0}_{N \times (M-N)} \\
\hat{\bV}_Y^{(1)} & -\hat{\bV}_Y^{(1)} &  \bV_Y^{(2)}
\end{array}
\right] \left[
\begin{array}{ccc}
\bGam_N & \mathbf{0} & \mathbf{0}\\
\mathbf{0} & -\bGam_N &  \mathbf{0}\\
\mathbf{0} & \mathbf{0} & \mathbf{0} 
\end{array}
\right]  \left[
\begin{array}{ccc}
\hat{\bU}_Y & \hat{\bU}_Y &  \mathbf{0}_{N \times (M-N)} \\
\hat{\bV}_Y^{(1)} & -\hat{\bV}_Y^{(1)} &  \bV_Y^{(2)}
\end{array}
\right]^T
\label{eigen-cY}
\end{equation}
\end{minipage}
\end{figure*}
Construct the symmetric matrix $\cY \in \bR^{(N+M) \times (N+M)}$ from the matrix $\bY$,
\begin{equation*}
    \cY = \left[
\begin{array}{cc}
\mathbf{0}_{N\times N} & \bY \\
\bY^T & \mathbf{0}_{M\times M}
\end{array}
\right]
\end{equation*}
By Theorem 7.3.3 in \cite{horn2012matrix}, $\cY$ has the eigen-decomposition given in \eqref{eigen-cY}, with $\bV_Y = \left[
\begin{array}{cc}
\bV_Y^{(1)} & \bV_Y^{(2)}
\end{array}
\right]$ in which $\bV_Y^{(1)} \in \bR^{M \times N}$. And, $\hat{\bV}_Y^{(1)} = \frac{1}{\sqrt{2}} \bV_Y^{(1)}$, $\hat{\bU}_Y= \frac{1}{\sqrt{2}} \bU_Y$. Eigenvalues of $\cY$ are signed singular values of $\bY$, therefore the limiting eigenvalue distribution of $\cY$ (ignoring zero eigenvalues) is the same as the limiting symmetrized singular value distribution of $\bY$.

Define the resolvent of $\cY$ 
\begin{equation}
    \bG_{\mathcal{Y}}(z) = \left[
\begin{array}{c}
z \bI - \cY 
\end{array}
\right]^{-1}
\end{equation}
and denote the eigenvectors of $\cY$ by $\bw_i \in \bR^{M+N}$, $i = 1, \dots,M+N$. For $z = x - \ci y$ with $x \in \bR$ and $y \gg \frac{1}{N}$, we have:
\begin{equation}
     \bG_{\mathcal{Y}}(x- \ci y) = \sum_{k=1}^{2N} \frac{x + \ci y }{(x - \tilde{\gamma}_k)^2+y^2} \bw_k \bw^T_k 
\end{equation}
where $\tilde{\gamma}_k$ are the eigenvalues of $\cY$, which are in fact the (signed) singular values of $\bY$, $\tilde{\gamma}_1 = \gamma_1, \hdots, \tilde{\gamma}_N=\gamma_N, \tilde{\gamma}_{N+1}= - \gamma_1, \hdots, \tilde{\gamma}_{2N}=-\gamma_N$.

Define set of vectors $\br_i, \bl_i \in \bR^{N+M}$ for $i=1,\dots,N$ as:
\begin{equation*}
    \br_i = \left[
\begin{array}{c}
\mathbf{0}_N \\
\bs_i^{(r)}
\end{array}
\right] \hspace{1cm}
\bl_i =\left[
\begin{array}{c}
\bs_i^{(l)} \\
\mathbf{0}_M
\end{array}
\right]
\end{equation*}

We have
\begin{equation*}
\begin{split}
    \br_j^T \big( {\rm Im}\, \bG_{\mathcal{Y}}(x - \ci y) \big) \bl_j &= \sum_{k=1}^{2N} \frac{y }{(x - \tilde{\gamma}_k)^2+y^2} (\br_j^T\bw_k) (\bw^T_k \bl_j) \\
    &= \frac{1}{2} \frac{1}{N} \sum_{k=1}^{2 N} \frac{y }{(x - \tilde{\gamma}_k)^2+y^2} \tilde{O}_{k,j}
\end{split}
\end{equation*}
with $\tilde{O}_{k,j} = (-1)^{\mathbb{I}(k>N)} N \big(  \bu_k^T \bs_j^{(l)}\big) \big(  \bv_k^T \bs_j^{(r)} \big)$. We assume that this quantity  concentrates on its mean. It turns out that this mean is a function of the singular values $\gamma_k, \sigma_j$ denoted by $O(\gamma_k,\sigma_j) \equiv  \big \langle \tilde{O}_{k,j} \big \rangle$. Taking the limit $N \to \infty$, we find
\begin{equation*}
    \br_j^T \big( {\rm Im}\, \bG_{\mathcal{Y}}(x - \ci y) \big) \bl_j \approx \int_\bR \frac{y}{(x - t)^2+ y^2} O(t,\sigma_j) \bar{\mu}_{Y}(t) \, dt
\end{equation*}
where $O(t,\sigma_j)$ is extended (continuously) to arbitrary values inside the support of $\bar{\mu}_Y$ (the symmetrized limiting singular value distribution of $\bY$).

Sending $y \to 0$, we find 
\begin{equation}
    \br_j^T \big( {\rm Im}\, \bG_{\mathcal{Y}}(x- \ci y) \big) \bl_j \approx \pi \bar{\mu}_{Y}(x) O(x, \sigma_j) 
    \label{resolvent-overlap}
\end{equation}
Eq. \eqref{resolvent-overlap} is important because it enables us to investigate the overlap through the resolvent of $\cY$. In the next section, we derive a relation between this resolvent and the signal $\bS$ which will allow us to find the optimal singular values $\xi^*_i$'s in terms of the singular values of the observation matrix $\bY$. 

\subsection{Resolvent Relation}
In Appendix \ref{RIE-der}, we show that we have relation \eqref{resolvent-relation} for the resolvent $\bG_{\mathcal{Y}}(z)$, in which $\langle . \rangle$ is the expectation w.r.t. the singular vectors of $\bY$, and $\bG_{S^T S}$ is the resolvent matrix of $\bS^T \bS$.
\begin{figure*}[t]
\centering
\begin{minipage}{\textwidth}
\begin{equation}
\begin{split}
\langle \bG_{\mathcal{Y}}(z) \rangle &= \Bigg \langle \left[
\begin{array}{cc}
z^{-1} \bI_N + z^{-1} \bY \bG_{Y^TY}(z^2) \bY^T & \bY \bG_{Y^TY}(z^2) \\
\bG_{Y^TY}(z^2) \bY^T & z \bG_{Y^TY}(z^2)
\end{array}
\right] \Bigg \rangle\\
&\approx\left[
\begin{array}{cc}
(z-\zeta_a^*)^{-1} \bI_N + (z-\zeta_a^*)^{-1} \bS \bG_{S^T S} \big((z-\zeta_b^*)(z - \zeta_a^*)\big) \bS^T &   \bS \bG_{S^T S} \big((z-\zeta_b^*)(z - \zeta_a^*)\big)  \\
\bG_{S^T S}\big((z-\zeta_b^*)(z - \zeta_a^*)\big) \bS^T & (z - \zeta_a^*) \bG_{S^T S} \big((z-\zeta_b^*)(z - \zeta_a^*)\big)
\end{array}
\right]
\end{split}
\label{resolvent-relation}
\end{equation}

\begin{equation}
    \begin{cases}
    \zeta^*_a = z\frac{Z(z)}{\mathcal{M}_{\mu_Y} \big( \frac{1}{z^2} \big) +1} \\
    \zeta^*_b  = \alpha z\frac{Z(z)}{\alpha \mathcal{M}_{\mu_Y}  \big( \frac{1}{z^2} \big) +1}
    \end{cases}
    \hspace{20 pt}
    Z(z) = \mathcal{C}^{(\alpha)}_{\mu_Z}\bigg(\frac{1}{z^2} T^{(\alpha)} \Big( \mathcal{M}_{\mu_Y}  \big( \frac{1}{z^2} \big) \Big)\bigg)
    \label{zeta_sol}
\end{equation}
\begin{align}
        \hspace{-1cm}\xi_i^* &= \frac{1}{\pi \bar{\mu}_{Y}(\gamma_i)} {\rm Im} \,\Bigg[  \gamma_i \mathcal{G}_{\bar{\mu}_Y}(\gamma_i - \ci 0^+) - \mathcal{C}^{(\alpha)}_{\mu_Z}\bigg(\frac{1}{\gamma_i}  \mathcal{G}_{\bar{\mu}_Y}(\gamma_i - \ci 0^+) \Big(1 - \alpha + \alpha \gamma_i  \mathcal{G}_{\bar{\mu}_Y}(\gamma_i - \ci 0^+) \Big) \bigg) \Bigg] \label{optimal-sv-final} \\
        &= \gamma_i - \frac{1}{\pi \bar{\mu}_{Y}(\gamma_i)} {\rm Im} \, \mathcal{C}^{(\alpha)}_{\mu_Z}\bigg( \frac{1- \alpha}{\gamma_i} \pi \sH [\bar{\mu}_{Y}](\gamma_i) + \alpha \big(\pi \sH [\bar{\mu}_{Y}](\gamma_i)\big)^2  - \alpha  \big( \pi \bar{\mu}_{Y}(\gamma_i)\big)^2 + \ci \pi \bar{\mu}_{Y}(\gamma_i) \big(\frac{1-\alpha}{\gamma_i} + 2 \alpha \pi \sH [\bar{\mu}_{Y}](\gamma_i) \big) \bigg) \notag
\end{align}
\vspace{-3mm}
\medskip
\hrule
\end{minipage}
\vspace{-5mm}
\end{figure*}
As a sanity check, by considering the normalized trace of the first block on both sides of \eqref{resolvent-relation}, one can recover the free rectangular addition formula $\mathcal{C}^{(\alpha)}_{\mu_S}(u) + \mathcal{C}^{(\alpha)}_{\mu_Z}(u) =  \mathcal{C}^{(\alpha)}_{\mu_Y}(u) $ for $u = \frac{1}{z^2} T^{(\alpha)} \Big( \mathcal{M}_{\mu_Y}  \big( \frac{1}{z^2} \big) \Big)$ (see Appendix \ref{free-add-conv}).

\subsection{Overlap and Optimal Singular Values}
From the lower left block of \eqref{resolvent-relation}, we get:
\begin{equation*}
\begin{split}
    \br_j^T \, \bG_{\mathcal{Y}}(z) \, \bl_j &= {\bs_j^{(r)}}^T \bG_{S^T S}\big((z-\zeta_b^*)(z - \zeta_a^*)\big) \bS^T \bs_j^{(l)} \\
&= \frac{\sigma_j}{(z-\zeta_b^*)(z - \zeta_a^*) - {\sigma_j}^2}
\end{split}
\end{equation*}
and using \eqref{resolvent-overlap}, we find:
\begin{equation}
    \hspace{-5mm}O(\gamma, \sigma) = \frac{1}{\pi \bar{\mu}_{Y}(\gamma)}  \lim_{z \to \gamma - \ci 0^+} {\rm Im}\, \frac{\sigma}{(z-\zeta_b^*)(z - \zeta_a^*) - {\sigma}^2}
    \label{Overlap-eq}
\end{equation}
where $\sigma$ is in the support of the limiting singular value distribution of $\bS$, $\mu_{S}$. In Fig. \ref{fig:Overlap} we illustrate on an example that theoretical predictions \eqref{Overlap-eq} are in good agreement with numerical simulations.
\begin{SCfigure}
  \centering
   \input{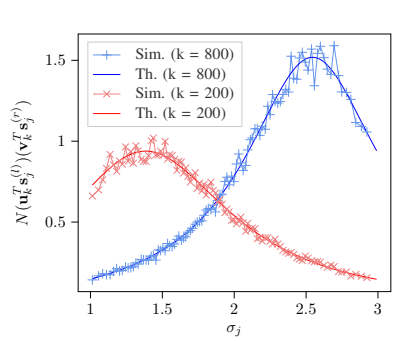}
   \caption{\small Computation of the rescaled overlap. Both $\bS$ and $\bZ$ are $N \times M$ matrices with i.i.d. Gaussian entries of variance $1/N$, and $N/M = 1/4$. The simulation results are average of 1000 experiments with fixed $\bS$, and $N = 1000, M =4000$. Some of the simulation points are dropped for clarity.}
   \label{fig:Overlap}
\end{SCfigure}

The optimal estimator for singular values reads:
\vspace{-1mm}
\begin{align}\label{optxi}
    \xi_i^* &= \frac{1}{N} \sum_{j=1}^N \sigma_j O(\gamma_i, \sigma_j) \approx \int t O(\gamma_i,t) \mu_{S}(t) \, dt \\
    &\hspace{-11pt}\approx  \frac{1}{\pi \bar{\mu}_{Y}(\gamma_i)} \lim_{z \to \gamma_i - \ci 0^+}  {\rm Im}\,\int \frac{t^2}{(z-\zeta_b^*)(z - \zeta_a^*) - t^2}  \mu_{S}(t) \, dt \notag \\
    &\hspace{-11pt}\approx  \frac{1}{N \pi \bar{\mu}_{Y}(\lambda_i)} \lim_{z \to \gamma_i - \ci 0^+}  {\rm Im} \Tr \bS \bG_{S^TS}\big((z-\zeta_b^*)(z - \zeta_a^*)\big)\bS^T \notag
\end{align}
Comparing the left upper blocks in the first and second lines of \eqref{resolvent-relation} we find
\begin{equation}
\begin{split}
    \bS \bG_{S^TS} &\big((z-\zeta_b^*)(z - \zeta_a^*)\big)\bS^T \\
    &= \Big\langle - \frac{\zeta_a^*}{z} \bI_N + \big( 1 - \frac{\zeta_a^*}{z} \big) \bY \bG_{Y^TY}(z^2)\bY^T \Big \rangle
\end{split}
\label{optimal-sv-eq1}
\end{equation}
Trace of the rhs of \eqref{optimal-sv-eq1} is (with multiplication by $1/N$)
\begin{equation*}
    \begin{split}
\frac{1}{N} \sum_{k=1}^N &\Big[\frac{{\gamma_k}^2}{z^2 - {\gamma_k}^2 }\big( 1 - \frac{\zeta_a^*}{z} \big) - \frac{\zeta_a^*}{z} \Big] \\
        &=- \frac{\zeta_a^*}{z} \frac{1}{N} \sum_{k=1}^N \big[\frac{{\gamma_k}^2}{z^2 - {\gamma_k}^2 }+1 \big] +  \frac{1}{N} \sum_{k=1}^N \frac{{\gamma_k}^2}{z^2 - {\gamma_k}^2 } \\
        &\approx - \zeta_a^* z \mathcal{G}_{\rho_Y}(z^2) +  \mathcal{M}_{\mu_Y}  \big( \frac{1}{z^2} \big) \\
        &= - \zeta_a^* z \mathcal{G}_{\rho_Y}(z^2) +  z^2 \mathcal{G}_{\rho_Y}(z^2) - 1
    \end{split}
\end{equation*}
The rhs can be expressed in terms of the symmetrized limiting spectral distribution of $\bY$. Indeed if we denote the Stieltjes of $\bar{\mu}_{Y}$ by $\mathcal{G}_{\bar{\mu}_{Y}}(z)$, using the relation $z \mathcal{G}_{\rho_Y}(z^2) = \mathcal{G}_{\bar{\mu}_Y}(z)$, the above trace implies with \eqref{optimal-sv-eq1}:
\begin{equation*}
    \frac{1}{N} \Tr \bS \bG_{S^TS}\big((z-\zeta_b^*)(z - \zeta_a^*)\big)\bS = - \zeta_a^* \mathcal{G}_{\bar{\mu}_Y}(z) +  z \mathcal{G}_{\bar{\mu}_Y}(z) - 1
\end{equation*}
Moreover $\zeta_a^*$ in \eqref{zeta_sol} can be written as,
\begin{equation}
    \begin{split}
        \zeta_a^* = \frac{1}{\mathcal{G}_{\bar{\mu}_Y}(z)} \mathcal{C}^{(\alpha)}_{\mu_Z}\bigg(\frac{1}{z}  \mathcal{G}_{\bar{\mu}_Y}(z) \Big(1 - \alpha + \alpha z \mathcal{G}_{\bar{\mu}_Y}(z) \Big) \bigg)
    \end{split}
\end{equation}
Replacing these results in \eqref{optxi} we easily deduce \eqref{optimal-sv-final} for the optimal singular values of the RIE.



\section{Computation of MMSE}\label{MMSE-computation}
First, we show that the posterior mean estimator is a RIE. 
\begin{equation}
   \bE [ \bS | \bY ]= \frac{1}{Z(\bY)} \int d\bX \, P_{S}(\bX) \bX P_Z \big( \bY - \sqrt{\lambda}\bX \big)
\end{equation}
with $Z(\bY)$ the normalizing constant. By the same calculations as below, one can see that $Z(\bY) = Z (\bU \bY \bV^T)$.

By rotation invariance of $P_{S}(\bX)$ under any orthogonal transformation $\bX \to \bU \bX \bV^T$ with Jacobian $\vert{\rm det} \bU\vert = \vert{\rm det} \bV \vert = 1$ we have 
\begin{align}
        &\bE [ \bS | \bU \bY \bV^T ] \hspace{-2pt}=\hspace{-2pt} \frac{1}{Z(\bY)} \hspace{-2pt} \int d\bX  P_{S}(\bX) \bX P_Z \big( \bU \bY \bV^T \hspace{-1.6pt}- \sqrt{\lambda}\bX \big) \notag \\
        &= \frac{1}{Z(\bY)} \int d\bX \, P_{S}(\bX) \bU \bX \bV^T P_Z \big( \bU \bY \bV^T \hspace{-1.6pt} - \sqrt{\lambda}\bU \bX \bV^T \big) \notag  \\
        &= \bU  \Big\{ \frac{1}{Z(\bY)}\int d\bX \,P_{S}(\bX) \bX P_Z \big( \bY - \sqrt{\lambda}\bX \big) \Big\} \bV^T \notag  \\
        &= \bU \bE [ \bS | \bY ] \bV^T
\end{align}
On the one hand, the posterior mean estimator achieves the MMSE {\it and} is RIE. On the other hand, the estimator in \eqref{rect-RIE} is conjectured to have the minimum MSE among the RIE class. Therefore, the MSE of the estimator \eqref{rect-RIE} equals the MMSE.

From \eqref{MSE-expansion}, \eqref{optimal-sv}, the MSE of the RIE in \eqref{rect-RIE} is $\frac{1}{N} \sum_{i =1}^N \sigma_i^2 - \frac{1}{N} \sum_{i =1}^N {\xi^*_i}^2$. Assuming that the MMSE concentrates in the limit $N \to \infty$, we obtain \eqref{MMSE-eq}.

\subsection{Gaussian Noise}
From \eqref{MMSE-eq}, \eqref{G-RIE}, to compute the MMSE, we need to compute the following expectation:
\begin{equation*}
    \int \bigg( x - \frac{1-\alpha}{\alpha} \frac{1}{x} - 2 \pi \sH [\bar{\mu}_{Y}](x) \bigg)^2 \, \mu_Y(x) \, dx
\end{equation*}
In Appendix \ref{com-G-MMSE}, using properties of the Hilbert transform we show this integral equals:
\begin{equation}
    \int x^2 \mu_Y(x) \, dx + \big(\frac{1}{\alpha} -1 \big)^2 \int \frac{\mu_Y(x)}{x^2}  dx + \frac{\pi^2}{3} \int {\mu_Y(x)}^3 \, dx - \frac{2}{\alpha}
    \label{expect-xi*}
\end{equation}
By independence of $\bS$ and $\bZ$, the second moment of $\mu_Y(x)$ can be expressed as $\int x^2 \mu_Y(x) \, dx =  \lambda \int x^2 \mu_S(x) \, dx + \frac{1}{\alpha}$.
Putting these relations together, we deduce (for Gaussian noise):
\begin{equation*}
\begin{split}
    \int {\xi^*(x)}^2 \mu_Y(x) \, dx &= \int x^2 \mu_S(x) \, dx - \frac{1}{\lambda} \Big[ \frac{1}{\alpha} \\
    &\hspace{-30pt}- \big(\frac{1}{\alpha} -1 \big)^2 \int \frac{\mu_Y(x)}{x^2}\, dx 
        - \frac{\pi^2}{3} \int {\mu_Y(x)}^3 \, dx \Big]
\end{split}
\end{equation*}
Replacing this identity in \eqref{MMSE-eq}, we get \eqref{G-MMSE-eq}.
\section{Numerical Results}\label{numericals}
We consider two cases of signal priors under two different noise distributions. We look at Gaussian and sparse signals and Gaussian noise ($\bZ$ has i.i.d. Gaussian entries with variance $\frac{1}{N}$ and $\mu_Z$ is the Marchenko-Pastur distribution) and the uniform noise with $\mu_Z$  uniform distribution on $[0,2]$, $\mathcal{U}_{[0,2]}$. For the uniform distribution with $\alpha=1$, the rectangular R-transform is $\mathcal{C}_{\mathcal{U}_{[0,2]}}^{(1)} (z) = 2 \sqrt{z} \coth \big(2\sqrt{z}\big) - 1$. To compute the RIE, $\bar{\mu}_Y(\gamma)$ and $\sH [\bar{\mu}_{Y}](\gamma)$ are computed from an estimation of the Stieltjes transform of $\bar{\mu}_Y$, which is obtained numerically using the Cauchy kernel (see section 19.5.2 in \cite{potters2020first}). Note that, the knowledge of $\mu_S$ is not needed in \eqref{rect-RIE}, and thus this approach for constructing the RIE provides an estimation \textit{algorithm}.

\subsection{Gausssian Signal}
Let $\bS$ a  matrix with i.i.d. Gaussian entries with variance $\frac{1}{N}$.
\subsubsection{Gaussian Noise}
In this case, for any $N,M$ each entry of $\bY$ can be viewed as an independent scalar AWGN channel. For this scalar channel, the MMSE equals $\frac{1}{N} \frac{1}{1 + \lambda}$ \cite{guo2005mutual}. Therefore, the (normalized) MMSE of the matrix problem is $\frac{M}{N} \frac{1}{1 + \lambda}\to \frac{1}{\alpha} \frac{1}{1 + \lambda}$ for $N \to \infty$. $\mu_Y$ is the Marchenko-Pastur (MP) law rescaled with $\sqrt{\lambda + 1}$ and the MMSE in \eqref{G-MMSE-eq} can be evaluated analytically (by \textit{Mathematica} \cite{Mathematica}) to be equal to $\frac{1}{\alpha} \frac{1}{1 + \lambda}$. 
In Fig. \ref{G-G}, MSE of RIE is compared to the theoretical MMSE for $\alpha = 1$. Note that, for this example, the Hilbert transform used in RIE is the exact Hilbert transform of MP law rescaled with $\sqrt{\lambda + 1}$.

\subsubsection{Uniform Noise}
In Fig. \ref{G-U}, we compare the MSE of the RIE \eqref{rect-RIE} with the MSE of the oracle estimator \eqref{optimal-sv}. We see that, except in the relatively low SNR, RIE has the same MSE as the oracle estimator. We believe that, the mismatch in the low-SNR regime is due to the inaccuracy of numerical approximation of the Stieltjes transform of $\bar{\mu}_Y$.
\begin{figure}
\centering
    \begin{subfigure}[b]{0.22\textwidth}
        \centering
        \input{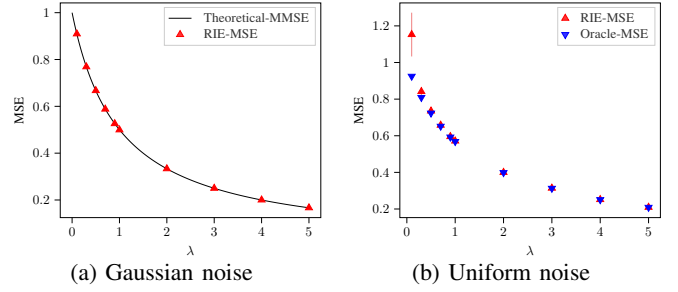}
        \vspace{-6mm}
        \caption{{\small Gaussian noise}}
        \label{G-G}
    \end{subfigure}
    \quad
    \begin{subfigure}[b]{0.22\textwidth}
        \centering
\begin{tikzpicture}[scale=0.5]

\definecolor{color0}{rgb}{0.941176470588235,0.501960784313725,0.501960784313725}
\definecolor{color1}{rgb}{0.67843137254902,0.847058823529412,0.901960784313726}

\begin{axis}[
legend cell align={left},
legend style={fill opacity=0.8, draw opacity=1, text opacity=1, draw=white!80!black},
tick align=outside,
tick pos=left,
x grid style={white!69.0196078431373!black},
xlabel={\(\displaystyle \lambda\)},
xmin=-0.145, xmax=5.245,
xtick style={color=black},
y grid style={white!69.0196078431373!black},
ylabel={MSE},
ymin=0.152989628270105, ymax=1.32539033414586,
ytick style={color=black}
]
\path [draw=color0, semithick]
(axis cs:0.1,1.03303033145219)
--(axis cs:0.1,1.27209939296969);

\path [draw=color0, semithick]
(axis cs:0.3,0.81987575197821)
--(axis cs:0.3,0.863033131115337);

\path [draw=color0, semithick]
(axis cs:0.5,0.726273947666152)
--(axis cs:0.5,0.743930384552147);

\path [draw=color0, semithick]
(axis cs:0.7,0.648403717688098)
--(axis cs:0.7,0.666974801994573);

\path [draw=color0, semithick]
(axis cs:0.9,0.5901955957818)
--(axis cs:0.9,0.600516926620797);

\path [draw=color0, semithick]
(axis cs:1,0.565385729356802)
--(axis cs:1,0.577083622597401);

\path [draw=color0, semithick]
(axis cs:2,0.393250011156362)
--(axis cs:2,0.406413307591626);

\path [draw=color0, semithick]
(axis cs:3,0.308871790800908)
--(axis cs:3,0.318785866515836);

\path [draw=color0, semithick]
(axis cs:4,0.247606338085179)
--(axis cs:4,0.255309250887774);

\path [draw=color0, semithick]
(axis cs:5,0.206662737482111)
--(axis cs:5,0.212248083249741);

\path [draw=color1, semithick]
(axis cs:0.1,0.923737564407104)
--(axis cs:0.1,0.925920607197669);

\path [draw=color1, semithick]
(axis cs:0.3,0.804393241100484)
--(axis cs:0.3,0.81386038120777);

\path [draw=color1, semithick]
(axis cs:0.5,0.720281014015007)
--(axis cs:0.5,0.726999508933095);

\path [draw=color1, semithick]
(axis cs:0.7,0.645750561762231)
--(axis cs:0.7,0.657776341270898);

\path [draw=color1, semithick]
(axis cs:0.9,0.588136725710247)
--(axis cs:0.9,0.597341366641143);

\path [draw=color1, semithick]
(axis cs:1,0.563689801272909)
--(axis cs:1,0.573432903917115);

\path [draw=color1, semithick]
(axis cs:2,0.392454165647133)
--(axis cs:2,0.404943625182792);

\path [draw=color1, semithick]
(axis cs:3,0.308128156083209)
--(axis cs:3,0.317630266805713);

\path [draw=color1, semithick]
(axis cs:4,0.247104440654114)
--(axis cs:4,0.254713902511564);

\path [draw=color1, semithick]
(axis cs:5,0.206280569446275)
--(axis cs:5,0.211848418333924);

\addplot [semithick, red, mark=triangle*, mark size=3, mark options={solid}, only marks]
table {%
0.1 1.15256486221094
0.3 0.841454441546774
0.5 0.735102166109149
0.7 0.657689259841335
0.9 0.595356261201299
1 0.571234675977102
2 0.399831659373994
3 0.313828828658372
4 0.251457794486477
5 0.209455410365926
};
\addlegendentry{RIE-MSE}
\addplot [semithick, blue, mark=triangle*, mark size=3, mark options={solid,rotate=180}, only marks]
table {%
0.1 0.924829085802387
0.3 0.809126811154127
0.5 0.723640261474051
0.7 0.651763451516565
0.9 0.592739046175695
1 0.568561352595012
2 0.398698895414963
3 0.312879211444461
4 0.250909171582839
5 0.2090644938901
};
\addlegendentry{Oracle-MSE}
\end{axis}

\end{tikzpicture}
        \vspace{-6mm}
        \caption{{\small Uniform noise}}
        \label{G-U}
    \end{subfigure}
    \vspace{-1mm}
    \caption{{\small Gaussian Signal with $\alpha =1$. The RIE is applied to $N=1000, M =1000$, and the results are averaged over 10 runs (error bars might be invisible).}}
    \label{fig:G-MSE}
    \vspace{-3mm}
\end{figure}

\subsection{Sparse Signal}
In this case $\mu_S = p \delta_0 + (1-p)\delta_{+1}$ for $0 \leq p \leq 1$. In Fig. \ref{fig:S-MSE} MSE of RIE is compared to the MSE of oracle estimator \eqref{optimal-sv} with $\alpha = 1$ for $p=0.2,0.9$, for Gaussian noise and uniform noise. For the Gaussian noise setting, in the high-sparsity regime $p=0.9$, the MSE is close to the rank-one MMSE computed in \cite{miolane2017fundamental}.

\begin{figure}
\centering
    \begin{subfigure}[b]{0.22\textwidth}
        \centering
        \input{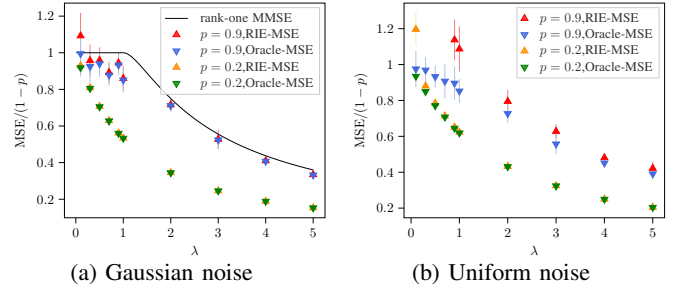}
        \vspace{-6mm}
        \caption{{\small Gaussian noise}}
        \label{S-G}
    \end{subfigure}
    \quad
    \begin{subfigure}[b]{0.22\textwidth}
        \centering
\begin{tikzpicture}[scale=0.5]

\definecolor{color0}{rgb}{0.941176470588235,0.501960784313725,0.501960784313725}
\definecolor{color1}{rgb}{0.690196078431373,0.768627450980392,0.870588235294118}
\definecolor{color2}{rgb}{1,0.894117647058824,0.709803921568627}
\definecolor{color3}{rgb}{0.564705882352941,0.933333333333333,0.564705882352941}
\definecolor{color4}{rgb}{0.254901960784314,0.411764705882353,0.882352941176471}
\definecolor{color5}{rgb}{1,0.549019607843137,0}

\begin{axis}[
legend cell align={left},
legend style={fill opacity=0.8, draw opacity=1, text opacity=1, draw=white!80!black},
tick align=outside,
tick pos=left,
x grid style={white!69.0196078431373!black},
xlabel={\(\displaystyle \lambda\)},
xmin=-0.145, xmax=5.245,
xtick style={color=black},
y grid style={white!69.0196078431373!black},
ylabel={${\rm MSE}/(1-p)$},
ymin=0.146236429442043, ymax=1.34251990027516,
ytick style={color=black}
]
\path [draw=color0, semithick]
(axis cs:0.9,1.02387726616109)
--(axis cs:0.9,1.25140669020073);

\path [draw=color0, semithick]
(axis cs:1,0.960974133336603)
--(axis cs:1,1.21251428319012);

\path [draw=color0, semithick]
(axis cs:2,0.732456990293317)
--(axis cs:2,0.857534529874897);

\path [draw=color0, semithick]
(axis cs:3,0.588609497965468)
--(axis cs:3,0.667310433542997);

\path [draw=color0, semithick]
(axis cs:4,0.45826747883989)
--(axis cs:4,0.503866889837449);

\path [draw=color0, semithick]
(axis cs:5,0.384802253093513)
--(axis cs:5,0.459026716703951);

\path [draw=color1, semithick]
(axis cs:0.1,0.874052944415051)
--(axis cs:0.1,1.07702872171826);

\path [draw=color1, semithick]
(axis cs:0.3,0.897480081036266)
--(axis cs:0.3,1.04067888883746);

\path [draw=color1, semithick]
(axis cs:0.5,0.870567693813406)
--(axis cs:0.5,0.9943377177662);

\path [draw=color1, semithick]
(axis cs:0.7,0.809696788023236)
--(axis cs:0.7,1.00460296459583);

\path [draw=color1, semithick]
(axis cs:0.9,0.802630087830347)
--(axis cs:0.9,0.989532774802562);

\path [draw=color1, semithick]
(axis cs:1,0.786429996679347)
--(axis cs:1,0.917975874709159);

\path [draw=color1, semithick]
(axis cs:2,0.676906750543878)
--(axis cs:2,0.776366944653252);

\path [draw=color1, semithick]
(axis cs:3,0.502417740807798)
--(axis cs:3,0.61245667573098);

\path [draw=color1, semithick]
(axis cs:4,0.425946508365075)
--(axis cs:4,0.476251148415853);

\path [draw=color1, semithick]
(axis cs:5,0.35953473558791)
--(axis cs:5,0.421759773293055);

\path [draw=color2, semithick]
(axis cs:0.1,1.10445388884857)
--(axis cs:0.1,1.28814337887366);

\path [draw=color2, semithick]
(axis cs:0.3,0.866194659960202)
--(axis cs:0.3,0.895237117064461);

\path [draw=color2, semithick]
(axis cs:0.5,0.770119230773777)
--(axis cs:0.5,0.798808841862118);

\path [draw=color2, semithick]
(axis cs:0.7,0.702002745495262)
--(axis cs:0.7,0.725221596605315);

\path [draw=color2, semithick]
(axis cs:0.9,0.63812408649737)
--(axis cs:0.9,0.664819098028423);

\path [draw=color2, semithick]
(axis cs:1,0.610853221177445)
--(axis cs:1,0.635562413321669);

\path [draw=color2, semithick]
(axis cs:2,0.422778606210807)
--(axis cs:2,0.442490538452124);

\path [draw=color2, semithick]
(axis cs:3,0.315758266323445)
--(axis cs:3,0.334188412729614);

\path [draw=color2, semithick]
(axis cs:4,0.245134080687788)
--(axis cs:4,0.254582673293705);

\path [draw=color2, semithick]
(axis cs:5,0.201054373013121)
--(axis cs:5,0.207720573938528);

\path [draw=color3, semithick]
(axis cs:0.1,0.923628106001416)
--(axis cs:0.1,0.94550405281281);

\path [draw=color3, semithick]
(axis cs:0.3,0.837135483381572)
--(axis cs:0.3,0.861543477223333);

\path [draw=color3, semithick]
(axis cs:0.5,0.758727490654344)
--(axis cs:0.5,0.782955519684444);

\path [draw=color3, semithick]
(axis cs:0.7,0.697962128197737)
--(axis cs:0.7,0.716429578583355);

\path [draw=color3, semithick]
(axis cs:0.9,0.632864347065319)
--(axis cs:0.9,0.657006167141108);

\path [draw=color3, semithick]
(axis cs:1,0.60812725719676)
--(axis cs:1,0.629556168085208);

\path [draw=color3, semithick]
(axis cs:2,0.421430755540588)
--(axis cs:2,0.440418274911856);

\path [draw=color3, semithick]
(axis cs:3,0.314959614876798)
--(axis cs:3,0.332723455784719);

\path [draw=color3, semithick]
(axis cs:4,0.244562182629203)
--(axis cs:4,0.253811752811542);

\path [draw=color3, semithick]
(axis cs:5,0.200612950843549)
--(axis cs:5,0.207113648184605);

\addplot [semithick, red, mark=triangle*, mark size=3, mark options={solid}, only marks]
table {%
0.9 1.13764197818091
1 1.08674420826336
2 0.794995760084107
3 0.627959965754233
4 0.481067184338669
5 0.421914484898732
};
\addlegendentry{$p=0.9$,RIE-MSE}
\addplot [semithick, color4, mark=triangle*, mark size=3, mark options={solid,rotate=180}, only marks]
table {%
0.1 0.975540833066656
0.3 0.969079484936865
0.5 0.932452705789803
0.7 0.907149876309535
0.9 0.896081431316454
1 0.852202935694253
2 0.726636847598565
3 0.557437208269389
4 0.451098828390464
5 0.390647254440482
};
\addlegendentry{$p=0.9$,Oracle-MSE}
\addplot [semithick, color5, mark=triangle*, mark size=3, mark options={solid}, only marks]
table {%
0.1 1.19629863386111
0.3 0.880715888512331
0.5 0.784464036317947
0.7 0.713612171050288
0.9 0.651471592262896
1 0.623207817249557
2 0.432634572331465
3 0.32497333952653
4 0.249858376990746
5 0.204387473475824
};
\addlegendentry{$p=0.2$,RIE-MSE}
\addplot [semithick, green!50.1960784313725!black, mark=triangle*, mark size=3, mark options={solid,rotate=180}, only marks]
table {%
0.1 0.934566079407113
0.3 0.849339480302453
0.5 0.770841505169394
0.7 0.707195853390546
0.9 0.644935257103213
1 0.618841712640984
2 0.430924515226222
3 0.323841535330759
4 0.249186967720373
5 0.203863299514077
};
\addlegendentry{$p=0.2$,Oracle-MSE}
\end{axis}

\end{tikzpicture}
        \vspace{-6mm}
        \caption{{\small Uniform noise}}
        \label{S-U}
    \end{subfigure}
    \vspace{-1mm}
    \caption{{\small Sparse Signal with $\alpha =1$. MSE is normalized by the norm of the signal, $1-p$. The RIE is applied to $N=1000, M =1000$, and the results are averaged over 10 runs (error bars might be invisible).}}
    \label{fig:S-MSE}
    \vspace{-5.5mm}
\end{figure}


\section*{Acknowledgment}
The work of F. P has been supported by the SNSF grant no 200021-204119.
\clearpage
\bibliographystyle{IEEEtran}
\bibliography{References}

\clearpage
\appendices

\section{Derivation of the resolvent relation }\label{RIE-der}
From \eqref{RIE-observ-model}, we have
\begin{equation}
\begin{split}
    \cY &=  \left[
\begin{array}{cc}
\mathbf{0} & \bS \\
\bS^T & \mathbf{0}
\end{array}
\right] + \left[
\begin{array}{cc}
\bU & \mathbf{0}  \\
\mathbf{0} & \bV
\end{array}
\right] \left[
\begin{array}{cc}
\mathbf{0} & \bZ \\
\bZ^T & \mathbf{0}
\end{array}
\right] \left[
\begin{array}{cc}
\bU^T & \mathbf{0}  \\
\mathbf{0} & \bV^T
\end{array}
\right] \\
&= \cS + \bO \cZ \bO^T
\end{split}
\label{RIE-1}
\end{equation}

Let $\bG(z) \equiv \bG_{\mathcal{Y}}(z)$. We can express entries of $\bG(z)$ in terms of a Gaussian integral. In large-$N$, we expect the resolvent can be studied through its ensemble average over matrices $\bU$, $\bV$. Using the replica trick we have: (for more details we refer the reader to Appendix B-A in \cite{bun2016rotational})
\begin{equation*}
    \begin{split}
        &\langle  \bG_{ij}(z) \rangle = \lim_{n \to 0}  \int \bigg( \prod_{k=1}^{M+N} \prod_{\tau =1}^n d \eta_k^{\tau}  \bigg) \, \eta^1_i \eta^1_j \\
        & \times \bigg\langle \exp \Big\{ -\frac{1}{2} \sum_{\tau=1}^n \sum_{k,l=1}^{M+N} \eta_k^{\tau} \big( z \delta_{kl} - \cY_{kl} \big) \eta_l^{\tau} \Big\}  \bigg\rangle
    \end{split}
\end{equation*}
For the expression in the exponent, we have:
\begin{equation}
    \begin{split}
        &\hspace{-40pt}-\frac{1}{2} \sum_{\tau=1}^n \sum_{k,l=1}^{M+N} \eta_k^{\tau} \big( z \delta_{kl} - \cY_{kl} \big) \eta_l^{\tau}  \\
        &=-\frac{1}{2} \sum_{\tau=1}^n \sum_{k,l=1}^{M+N} \eta_k^{\tau} \big( z \delta_{kl} - \cS_{kl} \big) \eta_l^{\tau} \\
        &\hspace{20pt}+ \frac{1}{2} \sum_{\tau=1}^n \sum_{k,l=1}^{M+N} \eta_k^{\tau} (\bO \cZ \bO^T)_{kl} \eta_l^{\tau} 
    \end{split}
    \label{exp1}
\end{equation}
The first term in the RHS can be written as
\begin{equation}
    -\frac{1}{2} \sum_{\tau=1}^n {\bbeta^{\tau} }^T \big( z \bI_{N+M} - \cS \big) \bbeta^{\tau}
    \label{exp2}
\end{equation}

Given the structure \eqref{RIE-1} for $\bO \cZ \bO^T$, the second sum in \eqref{exp1} can be written as:
\begin{equation}
    \begin{split}
 \sum_{k=1}^N \sum_{l=N+1}^{M+N} &\eta_k^{\tau} \big(\bU \bZ \bV^T)_{k,l-N} \eta_l^{\tau}\\
 &+ \sum_{k=N+1}^{M+N} \sum_{l=1}^{N} \eta_k^{\tau} \big(\bV \bZ^T \bU^T)_{k-N,l} \eta_l^{\tau}
    \end{split}
    \label{exp3}
\end{equation}
Split each replica $\bbeta^{\tau}$ into two vectors $\ba^{\tau} \in \bR^N, \bb^{\tau} \in \bR^M$, $\bbeta^{\tau} = \left[
\begin{array}{c}
\ba^{\tau} \\
\bb^{\tau}
\end{array}
\right]$. The expression in \eqref{exp3} can be rewritten as $
 2 \Tr \bb^{\tau} {\ba^{\tau}}^T \bU \bZ \bV^T 
 $.
So,  we have:
\begin{equation}
    \begin{split}
        \langle  \bG_{ij}(z) \rangle &= \int \bigg( \prod_{k=1}^{M+N} \prod_{\tau =1}^n d \eta_k^{\tau}  \bigg) \, \eta^1_i \eta^1_j \\
        &\times  \exp \Big\{-\frac{1}{2} \sum_{\tau=1}^n {\bbeta^{\tau} }^T \big( z \bI_{N+M} - \cS \big) \bbeta^{\tau} \Big\}  \\
        &\times \bigg\langle \exp \Big\{ \sum_{\tau=1}^n \Tr \bb^{\tau} {\ba^{\tau}}^T \bU \bZ \bV^T  \Big\}  \bigg\rangle_{\bU,\bV}
    \end{split}
\end{equation}
The last term $\langle - \rangle_{\bU, \bV}$ is the definition of the rank-one rectangular spherical integral studied in \cite{benaych2011rectangular}.

For a rank-one matrix $\bm{E} \in \bR^{M \times N}$ with non-zero singular value $\theta$, and $\bM \in \bR^{N \times M}$ with limiting singular value distribution $\mu$, in the limit $N \to \infty$, we have
\begin{equation*}
    \begin{split}
        \frac{1}{N} \ln \bigg\langle \exp \Big\{ \sqrt{N M} \Tr \bm{E} \bU \bM \bV \Big\} \bigg\rangle &\approx \int_0^{\theta} \frac{\mathcal{C}_{\mu}^{(\alpha)}(t^2)}{t} \, dt \\
        &\hspace{-2cm}= \frac{1}{2} \int_0^{\theta^2} \frac{\mathcal{C}_{\mu}^{(\alpha)}(t)}{t} \, dt \equiv \frac{1}{2} \mathcal{W}_{\mu}(\theta^2)
    \end{split}
\end{equation*}
Generalizing this result to finite rank-n case we have:
\begin{equation*}
\begin{split}
    \bigg\langle & \exp \Big\{ \sum_{\tau=1}^n \Tr \bb^{\tau} {\ba^{\tau}}^T \bU \bZ \bV^T  \Big\}  \bigg\rangle_{\bU,\bV} \\
    &\approx \exp \Big\{ \frac{N}{2} \sum_{\tau=1}^n \mathcal{W}_{\mu_Z} \big(\frac{1}{N M} \| \ba^{\tau} \|^2 \| \bb^{\tau} \|^2 \big) \Big\}
\end{split}
\end{equation*}
where we used that for each replica, the singular value of $\bb^{\tau} {\ba^{\tau}}^T$ is $\| \ba^{\tau} \| \| \bb^{\tau} \|$. Although this generalization has not been studied before, but we believe that the same idea as in \cite{guionnet2022asymptotics} can be applied to show it holds.


Therefore, we find
\begin{equation}
    \begin{split}
        &\langle  \bG_{ij}(z) \rangle = \int \bigg( \prod_{k=1}^{M+N} \prod_{\tau =1}^n d \eta_k^{\tau}  \bigg) \, \eta^1_i \eta^1_j \\
        &\hspace{20pt} \times \exp \bigg\{ \sum_{\tau=1}^n \Big[  -\frac{1}{2}  {\bbeta^{\tau} }^T \big( z \bI_{N+M} - \cS \big) \bbeta^{\tau} \\
        &\hspace{30pt}+   \frac{N}{2} \mathcal{W}_{\mu_Z} \big(\frac{1}{N M} \| \ba^{\tau} \|^2 \| \bb^{\tau} \|^2 \big) \Big] \bigg\} 
    \end{split}
    \label{first-integral}
\end{equation}

Introducing delta functions $\delta \big(p_a^{\tau} - \frac{1}{N}\|\ba^{\tau}\|^2\big)$, $\delta \big(p_b^{\tau} - \frac{1}{M}\|\bb^{\tau}\|^2\big)$, and considering the Fourier transform $\delta \big(p_a^{\tau} - \frac{1}{N}\|\ba^{\tau}\|^2\big) \propto \int \, d \zeta_a^{\tau}  \exp  \Big\{ -\frac{N}{2} \zeta_a^{\tau} \big( p_a^{\tau} - \frac{1}{N}\|\ba^{\tau}\|^2 \big) \Big\}$, the integral in \eqref{first-integral} can be written as \eqref{eq15} (see \cite{bun2016rotational} for details).
\begin{figure*}[t]
\centering
\begin{minipage}{\textwidth}
\begin{equation}
    \begin{split}
        \langle  \bG_{ij}(z) \rangle &\propto  \int \int \bigg( \prod d p_a^{\tau} \, d p_b^{\tau} \,  d \zeta_a^{\tau} \,  d \zeta_b^{\tau}  \bigg) \exp \Big\{ \frac{N}{2} \sum_{\tau=1}^n  \big[  \mathcal{W}_{\mu_Z} ( p_a^{\tau} p_b^{\tau} ) - \zeta_a^{\tau} p_a^{\tau} - \frac{1}{\alpha} \zeta_b^{\tau} p_b^{\tau} \big] \Big\} \\
        &\times  \int \bigg( \prod_{k=1}^{M+N} \prod_{\tau =1}^n d \eta_k^{\tau}  \bigg) \, \eta^1_i \eta^1_j  \exp \bigg\{ \sum_{\tau=1}^n \Big[  -\frac{1}{2}  {\bbeta^{\tau} }^T \big( z \bI  - \cS \big) \bbeta^{\tau} + \frac{1}{2} \zeta_a^{\tau} \|\ba^{\tau}\|^2 + \frac{1}{2} \zeta_b^{\tau} \|\bb^{\tau}\|^2   \Big] \bigg\} 
    \end{split}
    \label{eq15}
\end{equation}
\begin{equation}
    {\bM^1}^{-1} = \left[
\begin{array}{cc}
(z-\zeta_a^1)^{-1} \bI_N + (z-\zeta_a^1)^{-1} \bS \bG_{S^T S} \big((z-\zeta_b^1)(z - \zeta_a^1)\big) \bS^T &  \bS \bG_{S^T S} \big((z-\zeta_b^1)(z - \zeta_a^1)\big)  \\
\bG_{S^T S}\big((z-\zeta_b^1)(z - \zeta_a^1)\big) \bS^T & (z - \zeta_a^1) \bG_{S^T S} \big((z-\zeta_b^1)(z - \zeta_a^1)\big)
\end{array}
\right]
\label{M-inverse}
\end{equation}
\begin{equation}
    \begin{split}
        \langle  \bG_{ij}(z) \rangle &\propto  \int \int \big({\bM^1}^{-1}\big)_{ij}   \exp \bigg\{ \frac{N}{2} \sum_{\tau=1}^n  \big[  \mathcal{W}_{\mu_Z} ( p_a^{\tau} p_b^{\tau} ) - \zeta_a^{\tau} p_a^{\tau} - \frac{1}{\alpha} \zeta_b^{\tau} p_b^{\tau} \big] \\
        &\hspace{0.5cm} - \sum_{\tau=1}^n \frac{1}{2} \Big[ (M-N) \ln (z-\zeta_b^{\tau}) + \sum_{k=1}^N \ln \big\{ (z-\zeta_a^{\tau})(z-\zeta_b^{\tau}) - {\sigma_k}^2 \big\} \Big] \bigg\} \bigg( \prod d p_a^{\tau} \, d p_b^{\tau} \,  d \zeta_a^{\tau} \,  d \zeta_b^{\tau}  \bigg)
    \end{split}
    \label{eq19}
\end{equation}
\vspace{-2mm}
\medskip
\hrule
\end{minipage}
\vspace{-5mm}
\end{figure*}
The second integral in \eqref{eq15} is a Gaussian integral with matrix
\begin{equation}
    \bM^{\tau} = \left[
\begin{array}{cc}
(z-\zeta_a^{\tau}) \bI_N & -\bS  \\
-\bS^T & (z-\zeta_b^{\tau})  \bI_M 
\end{array}
\right]
\label{matrix-M}
\end{equation}
Using the formula for determinant of block matrices, we have 
\begin{equation*}
\begin{split}
    \det \bM^{\tau} &= \det \big[ (z-\zeta_a^{\tau})\bI_N - (z-\zeta_b^{\tau})^{-1} \bS \bS^T \big] \det \big[ (z-\zeta_b^{\tau}) \bI_M \big] \\
    &= (z-\zeta_b^{\tau})^{M-N} \prod_{k=1}^N \big[ (z-\zeta_a^{\tau})(z-\zeta_b^{\tau}) - {\sigma_k}^2 \big]
\end{split}
\end{equation*}

Except for the first replica, the Gaussian integral is (up to constants):
\begin{equation*}
    \exp \bigg\{ -\frac{1}{2} \Big[ (M-N) \ln (z-\zeta_b^{\tau}) + \sum_{k=1}^N \ln \big\{ (z-\zeta_a^{\tau})(z-\zeta_b^{\tau}) - {\sigma_k}^2 \big\} \Big] \bigg\}
\end{equation*}
And, the integral for the first replica is the above expression multiplied by $\big({\bM^1}^{-1}\big)_{ij}$.

By Proposition 2.8.7 \cite{bernstein2009matrix}), ${\bM^1}^{-1}$ can be written as \eqref{M-inverse}, with $\bG_{S^T S}$ the resolvent of the matrix $\bS^T \bS$. 

Putting all together, the integral in \eqref{eq15}, can be written as \eqref{eq19}. This integral can be evaluated using saddle-point method. By replica symmetric ansatz, at saddle-point we have that
\begin{equation*}
    p_a^{\tau} = p_a, \hspace{5pt}p_b^{\tau} = p_b, \hspace{5pt}\zeta_a^{\tau} = \zeta_a, \hspace{5pt}\zeta_b^{\tau} = \zeta_b
\end{equation*}

The extremum of the function in the exponent is at:
\begin{equation}
    \begin{cases}
    p_a^* = (z - \zeta^*_b) \mathcal{G}_{\rho_S}\big( (z-\zeta^*_a)(z-\zeta^*_b) \big) \\
    p_b^* = (1-\alpha ) \frac{1}{z - \zeta^*_b}+\alpha(z - \zeta^*_a) \mathcal{G}_{\rho_S}\big( (z-\zeta^*_a)(z-\zeta^*_b) \big)\\
    \zeta^*_a = \frac{\mathcal{C}^{(\alpha)}_{\mu_Z}( p_a^* p_b^* )} {p_a^*} \\
    \zeta^*_b = \alpha \frac{\mathcal{C}^{(\alpha)}_{\mu_Z}( p_a^* p_b^* )} {p_b^*}
    \end{cases}
\end{equation}
where $\mathcal{G}_{\rho_S}$ is the Stieltjes transform of the matrix $\bS \bS^T$, whose limiting eigenvalue distribution is the squared transform of the limiting singular value distribution of $\bS$.

To simplify the solution, we compute the normalized trace of both sides in \eqref{eq19}. First we compute the trace of the matrix $\bM^{-1}$ in \eqref{M-inverse} plugging $\zeta^*_a, \zeta^*_b$. Trace of the first block is:
\begin{equation}
    \begin{split}
 \frac{1}{N}& \frac{1}{z-\zeta_a^*} \sum_{k=1}^{N} \Big[ 1 +  \frac{{\sigma_k}^2}{(z-\zeta_b^*)(z - \zeta_a^*) - {\sigma_k}^2} \Big] \\
        &=\frac{1}{N} (z-\zeta_b^*) \sum_{k=1}^{N}  \frac{1}{(z-\zeta_b^*)(z - \zeta_a^*) - {\sigma_k}^2} \\
        &\approx (z-\zeta_b^*) \mathcal{G}_{\rho_S}\big( (z-\zeta_b^*)(z - \zeta_a^*) \big) \\
        &= p_a^*
    \end{split}
    \label{trof1stblock}
\end{equation}
Similarly, the trace of the last block can be computed to be $p^*_b$.

The matrix in the lhs is $\bG_{\mathcal{Y}}(z)$, which is 
\begin{equation}
    \begin{split}
        \bG_{\mathcal{Y}}(z) &= \big(z \bI - \cY \big)^{-1} \\
&\hspace{-10pt}= \left[
\begin{array}{cc}
z^{-1} \bI_N + z^{-1} \bY \bG_{Y^TY}(z^2) \bY^T & \bY \bG_{Y^TY}(z^2) \\
\bG_{Y^TY}(z^2) \bY^T & z \bG_{Y^TY}(z^2)
\end{array}
\right]
    \end{split}
    \label{resolvent_Y}
\end{equation}
Trace of the first block is:
\begin{equation}
\begin{split}
\frac{1}{N} \frac{1}{z} \sum_{k=1}^N \big[ 1 + \frac{{\gamma_k}^2}{z^2 - {\gamma_k}^2} \big] &= \frac{1}{N} z \sum_{k=1}^N \frac{1}{z^2 - {\gamma_k}^2} \\
    &\approx z \mathcal{G}_{\rho_Y}(z^2)
\end{split}
\label{trace of first LHS}
\end{equation}
Therefore, from \eqref{trof1stblock}, we find $p_a^* = z \mathcal{G}_{\rho_Y}(z^2)$.

Trace of the last block can be evaluated to be
$\alpha z \mathcal{G}_{\rho_Y}(z^2) + (1- \alpha ) \frac{1}{z}$. So, $p_b^* = \alpha z \mathcal{G}_{\rho_Y}(z^2) + (1- \alpha ) \frac{1}{z}$.






Thus, we find
\begin{equation}
    \begin{cases}
    p_a^* = z \mathcal{G}_{\rho_Y}(z^2) = \frac{1}{z} \mathcal{M}_{\mu_Y} \big( \frac{1}{z^2} \big) + \frac{1}{z} \\
    p_b^* = \alpha z \mathcal{G}_{\rho_Y}(z^2) + (1- \alpha ) \frac{1}{z}= \alpha \frac{1}{z} \mathcal{M}_{\mu_Y} \big( \frac{1}{z^2} \big) + \frac{1}{z}
    \end{cases}
\end{equation}
\begin{equation*}
    p_a^* p_b^* = \frac{1}{z^2} T^{(\alpha)} \Big( \mathcal{M}_{\mu_Y} \big( \frac{1}{z^2} \big) \Big)
\end{equation*}
which implies
\begin{equation}
    \begin{cases}
    \zeta^*_a = z\frac{\mathcal{C}^{(\alpha)}_{\mu_Z}\bigg(\frac{1}{z^2} T^{(\alpha)} \Big( \mathcal{M}_{\mu_Y} \big( \frac{1}{z^2} \big) \Big)\bigg)}{\mathcal{M}_{\mu_Y} \big( \frac{1}{z^2} \big) +1} \\
    \zeta^*_b =  \alpha z\frac{\mathcal{C}^{(\alpha)}_{\mu_Z}\bigg(\frac{1}{z^2} T^{(\alpha)} \Big( \mathcal{M}_{\mu_Y} \big( \frac{1}{z^2} \big) \Big)\bigg)}{\alpha \mathcal{M}_{\mu_Y} \big( \frac{1}{z^2} \big) +1}
    \end{cases}
    \label{zeta_sol-app}
\end{equation}

\section{Derivation of Rectangular Free Convolution}\label{free-add-conv}
Consider the normalized trace of the first block on each side in \eqref{resolvent-relation}. The trace of the first block of the lhs is computed in \eqref{trace of first LHS} which is $\frac{1}{z} \mathcal{M}_{\mu_Y}\big( \frac{1}{z^2} \big) + \frac{1}{z}$. The trace of the first block in rhs is computed in \eqref{trof1stblock} which is $(z-\zeta_b^*) \mathcal{G}_{\rho_S}\big( (z-\zeta_b^*)(z - \zeta_a^*) \big)$.
\begin{equation*}
    \begin{split}
        &\frac{1}{z} \mathcal{M}_{\mu_Y}\big( \frac{1}{z^2} \big) + \frac{1}{z} = (z-\zeta_b^*) \mathcal{G}_{\rho_S}\big( (z-\zeta_b^*)(z - \zeta_a^*) \big) \\
        &= (z-\zeta_b^*) \frac{1}{(z-\zeta_b^*)(z - \zeta_a^*)} \Big(\mathcal{M}_{\mu_S} \big( \frac{1}{(z-\zeta_b^*)(z - \zeta_a^*)} \big) + 1 \Big) \\
        &= \frac{1}{z - \zeta_a^*} \mathcal{M}_{\mu_S} \big( \frac{1}{(z-\zeta_b^*)(z - \zeta_a^*)} \big) + \frac{1}{z - \zeta_a^*}
    \end{split}
\end{equation*}
From which, we get:
\begin{equation*}
    (z - \zeta_a^*) \mathcal{M}_{\mu_Y}\big( \frac{1}{z^2} \big) + z - \zeta_a^* = z \mathcal{M}_{\mu_S} \big( \frac{1}{(z-\zeta_b^*)(z - \zeta_a^*)} \big) + z
\end{equation*}

Taking the $\zeta_a^*$ to the rhs, and plugging the expression for $\zeta_a^*$ from \eqref{zeta_sol-app}, after a bit of algebra we find:
\begin{align}
        &\mathcal{M}_{\mu_Y}\big( \frac{1}{z^2} \big) \label{eq27} \\
        &= \mathcal{M}_{\mu_S} \big( \frac{1}{(z-\zeta_b^*)(z - \zeta_a^*)} \big) + \mathcal{C}^{(\alpha)}_{\mu_Z}\bigg(\frac{1}{z^2} T^{(\alpha)} \Big( \mathcal{M}_{\mu_Y}\big( \frac{1}{z^2} \big) \Big)\bigg) \notag
\end{align}

Let $\frac{1}{z^2} T^{(\alpha)} \Big( \mathcal{M}_{\mu_Y}\big( \frac{1}{z^2} \big) \Big) = u$. Then, $\frac{1}{z^2} =  {\mathcal{H}_{\mu_Y}^{(\alpha)}}^{-1}(u)$. Moreover, from the definition one can see that $\mathcal{M}_{\mu_Y}\Big( {\mathcal{H}_{\mu_Y}^{(\alpha)}}^{-1}(u) \Big) = \mathcal{C}^{(\alpha)}_{\mu_Y} (u)$. So, \eqref{eq27} can be written as:
\begin{equation}
    \mathcal{C}^{(\alpha)}_{\mu_Y} (u) = \mathcal{M}_{\mu_S} \big( \frac{1}{(z-\zeta_b^*)(z - \zeta_a^*)} \big) + \mathcal{C}^{(\alpha)}_{\mu_Z}(u)
    \label{relation for u}
\end{equation}

From \eqref{zeta_sol-app},
\begin{equation*}
    \begin{split}
        &(z-\zeta_b^*)(z - \zeta_a^*) = z^2 \Big( 1 - \frac{\mathcal{C}^{(\alpha)}_{\mu_Z}(u)}{\mathcal{C}^{(\alpha)}_{\mu_Y}(u) + 1} \Big) \Big( 1 - \frac{\alpha \mathcal{C}^{(\alpha)}_{\mu_Z}(u)}{\alpha \mathcal{C}^{(\alpha)}_{\mu_Y}(u) + 1} \Big) \\
        &= \frac{z^2}{T^{(\alpha)}\big( \mathcal{C}^{(\alpha)}_{\mu_Y}(u)\big)} T^{(\alpha)} \big(  \mathcal{C}^{(\alpha)}_{\mu_Y}(u) -  \mathcal{C}^{(\alpha)}_{\mu_Z}(u) \big)
    \end{split}
\end{equation*}
The first factor, using the definition of $ \mathcal{C}^{(\alpha)}_{\mu_Y}(u)$, is:
\begin{equation*}
    \begin{split}
        \frac{z^2}{T^{(\alpha)}\big( \mathcal{C}^{(\alpha)}_{\mu_Y}(u)\big)} &= \frac{1}{\frac{1}{z^2}} \frac{1}{T^{(\alpha)}\big( \mathcal{C}^{(\alpha)}_{\mu_Y}(u)\big)} \\
        &= \frac{1}{{\mathcal{H}_{\mu_Y}^{(\alpha)}}^{-1}(u)} \frac{1}{\frac{u}{{\mathcal{H}_{\mu_Y}^{(\alpha)}}^{-1}(u)}} = \frac{1}{u}
    \end{split}
\end{equation*}

So, \eqref{relation for u} can be written as
\begin{equation*}
    \mathcal{C}^{(\alpha)}_{\mu_Y}(u) -  \mathcal{C}^{(\alpha)}_{\mu_Z}(u) = \mathcal{M}_{\mu_S} \Big( \frac{u}{T^{(\alpha)} \big(  \mathcal{C}^{(\alpha)}_{\mu_Y}(u) -  \mathcal{C}^{(\alpha)}_{\mu_Z}(u) \big)} \Big)
\end{equation*}

One can see that, if the limiting singular value distribution of $\bS$, is not $\delta(x)$, the unique solution to the equation $ \mathcal{M}_{\mu_S}  \big( \frac{u}{T^{(\alpha)}(x)} \big) = x$, is $x = \mathcal{C}^{(\alpha)}_{\mu_S}(u)$ (see lemma 4.2 in \cite{benaych2011rectangular} for a particular case). Therefore, we find:
\begin{equation}
    \mathcal{C}^{(\alpha)}_{\mu_Y}(u) -  \mathcal{C}^{(\alpha)}_{\mu_Z}(u) = \mathcal{C}^{(\alpha)}_{\mu_S}(u)
\end{equation}
as we expected.

\section{Computation of MMSE for the Gaussian Noise}\label{com-G-MMSE}
In this section, we compute the following integral
\begin{equation*}
    \int \big( x - \frac{1-\alpha}{\alpha} \frac{1}{x} - 2 \pi \sH [\bar{\mu}_{Y}](x) \big)^2 \, \mu_Y(x) \, dx
\end{equation*}
For simplicity we denote $\sH [\bar{\mu}_{Y}]$ by $\bar{\sH}$. Expanding the integrand, we find
\begin{equation}
\begin{split}
    x^2 + &\big( \frac{1-\alpha}{\alpha} \big)^2 \frac{1}{x^2} -  2 \frac{1-\alpha}{\alpha}\\
    &+ 4 \pi^2 \big(\bar{\sH}(x)\big)^2 - 4 \pi x \bar{\sH}(x) + 4 \pi \frac{1-\alpha}{\alpha} \frac{\bar{\sH}(x)}{x}
\end{split}
\label{expansion}
\end{equation}

To compute the expectation of the last three terms, we need the following properties of the Hilbert transform. 
\begin{lemma}\label{properties of Hilbert}
If $f : \bR \to \bR$ is compactly supported and
sufficiently regular, then one has the identities
\begin{equation}
   \int_\bR f(x)\big( \sH [f] (x) \big)^2 \, dx = \frac{1}{3} \int_\bR f^3(x) \, dx
   \label{Hilbert-1}
\end{equation}
\begin{equation}
   \int_\bR \sH [f] (x) x f(x) \, dx = \frac{1}{2 \pi} \Big( \int_\bR f(x) \, dx \Big)^2
   \label{Hilbert-2}
\end{equation}
\begin{equation}
   \int_\bR \frac{\sH [f] (x)}{x}  f(x) \, dx = -\frac{1}{2 \pi} \Big( \int_\bR \frac{f(x)}{x} \, dx \Big)^2
   \label{Hilbert-3}
\end{equation}
\label{Hilbert-iden}
\end{lemma}
\begin{proof}
The proof of the first two properties can be found in Lemma 3.1 of \cite{shlyakhtenko2020fractional}. To prove the last equality, we apply the same idea as in remark 3.2 of the above paper to write:
\begin{equation*}
    \begin{split}
        \int_\bR \frac{\sH [f] (x)}{x}  f(x) \, dx&= \frac{1}{2 \pi} \iint \big( \frac{1}{x} - \frac{1}{y} \big) \frac{1}{x-y} f(x) f(y) \, dx \, dy \\
        &= - \frac{1}{2 \pi} \iint  \frac{1}{xy} f(x) f(y) \, dx \, dy \\
        &= - \frac{1}{2 \pi} \Big( \int \frac{f(x)}{x} \, dx \Big)^2
    \end{split}
\end{equation*}
\end{proof}
Moreover, we use the fact that the Hilbert transform  of an even function is an odd function \cite{kschischang2006hilbert}, that $\bar{\sH}(x)$ is an odd function.

From \eqref{Hilbert-1}, we have:
\begin{equation}
     \int \big(\bar{\sH}(x)\big)^2 \bar{\mu}_Y(x)\, dx = \frac{1}{3} \int {\bar{\mu}_Y(x)}^3 \, dx
     \label{first-term}
\end{equation}
The lhs can be written as:
\begin{equation*}
\begin{split}
        & \frac{1}{2} \int_{\bR_+} \big(\bar{\sH}(x)\big)^2 \mu_Y(x)\, dx + \frac{1}{2} \int_{\bR_-} \big(\bar{\sH}(x)\big)^2 \mu_Y(-x)\, dx \\
        &= \frac{1}{2} \int_{\bR_+} \big(\bar{\sH}(x)\big)^2 \mu_Y(x)\, dx +  \frac{1}{2} \int_{\bR_+} \big(\bar{\sH}(-x)\big)^2 \mu_Y(x)\, dx \\
        &= \frac{1}{2} \int_{\bR_+} \big(\bar{\sH}(x)\big)^2 \mu_Y(x)\, dx +  \frac{1}{2} \int_{\bR_+} \big(\bar{\sH}(x)\big)^2 \mu_Y(x)\, dx \\
        &= \int_{\bR_+} \big(\bar{\sH}(x)\big)^2 \mu_Y(x)\, dx
\end{split}
\end{equation*}
The rhs in \eqref{first-term} equals $ \frac{1}{12} \int {\mu_Y(x)}^3 \, dx$. Therefore, the expectation of the fourth term in \eqref{expansion} is:
\begin{equation}
   4 \pi^2 \int  \big(\bar{\sH}(x)\big)^2 \mu_Y(x) \, dx = \frac{\pi^2}{3}  \int {\mu_Y(x)}^3 \, dx
   \label{expec1}
\end{equation}

Similarly, using symmetry properties of $\bar{H}, \bar{\mu}_Y$ we have that:
\begin{equation*}
    \int x \bar{H}(x) \bar{\mu}_Y(x) \, dx = \int x \bar{H}(x) \mu_Y(x) \, dx
\end{equation*}
Thus, by \eqref{Hilbert-2}, the expectation of the fifth term in \eqref{expansion} is:
\begin{equation}
    -4 \pi \int x \bar{H}(x) \mu_Y(x) \, dx = -2  \Big( \int_\bR \bar{\mu}_Y(x) \, dx \Big)^2 = -2
    \label{expec2}
\end{equation}

Again, by symmetry, we have that:
\begin{equation*}
    \int  \frac{\bar{H}(x)}{x} \bar{\mu}_Y(x) \, dx = \int \frac{\bar{H}(x)}{x}  \mu_Y(x) \, dx
\end{equation*}
Thus, by \eqref{Hilbert-3}, the expectation of the last term in \eqref{expansion} is:
\begin{equation}
    \int \frac{\bar{H}(x)}{x}  \mu_Y(x) \, dx =  \Big( \int_\bR \frac{\bar{\mu}_Y(x)}{x} \, dx \Big)^2 = 0
    \label{expec3}
\end{equation}
where we used that $\frac{\bar{\mu}_Y(x)}{x}$ is an odd function.

From \eqref{expansion}, \eqref{expec1}, \eqref{expec2}, \eqref{expec3}, we get \eqref{expect-xi*}.

\section{Proof Steps of Theorem 1}\label{Proof-Thm1}
In this section, we present the steps we need to prove theorem \ref{main-th}.We make the following assumption:

\begin{assumption}\label{assumptions on law}
 The empirical singular value distribution of $\bS$ converges almost surely weakly to a well-defined probability density function $\mu_S(x)$ with compact support in $[C_1, C_2]$ with $C_1, C_2 \in \bR_{\geq 0}$. Moreover, $\mu_S$ has bounded second moment $\int x^2 \, d \mu_S < \infty$, finite non-commutative entropy $\iint \ln | x - y |  d\mu_S(x) d\mu_S(y) > -\infty$, and $\int \ln |x| d\mu_S(x) > - \infty$.
\end{assumption}

\begin{assumption}\label{bounded-mom}
The second moment of $\mu_{S}^{(N)}$ is almost surely bounded.
\end{assumption}


 We start from the posterior distribution of the model \eqref{observation-matrix} which reads (up to some constants):
\begin{equation}
    \begin{split}
        P (\bX | \bY) &\propto e^{-\frac{N}{2} \| \bY - \sqrt{\lambda}\bX \|_F^2 } P_{S}(\bX) \\
        & \propto e^{N \Tr \big[ \sqrt{\lambda}\bX \bY^T - \frac{\lambda}{2} \bX \bX^T \big] } P_{S}(\bX)
    \end{split}
    \label{post-model1}
\end{equation}
The partition function is defined as the normalizing factor of the posterior distribution \eqref{post-model1}:
\begin{equation}
    Z(\bY) = \int d \bX e^{N \Tr \big[ \sqrt{\lambda}\bX \bY^T - \frac{\lambda}{2} \bX \bX^T \big] } P_{S}(\bX)
    \label{partition-function-def}
\end{equation}
and the free energy is defined as:
\begin{equation}
    F_N (\lambda) = -\frac{1}{M N} \bE_{Y} \big[ \ln Z(\bY) \big]
    \label{free-energy-def}
\end{equation}
One can easily see that the free energy is linked to the (average) mutual information via the relation:
\begin{equation*}
    \frac{1}{M N} \mathcal{I}_N(\bS;\bY) = F_N (\lambda) + \frac{\lambda}{2 M} \bE \big[ \Tr \bS \bS^T \big]
\end{equation*}
in which $\frac{1}{M} \bE \big[ \Tr \bS^T \bS \big]$ converges to the second moment of $\mu_S$ rescaled by the factor $\alpha$. Therefore, to prove theorem \ref{main-th}, it is enough to show that
\begin{equation*}
     \lim_{N \to \infty} F_N (\lambda)  = \frac{\lambda}{2} \alpha \int x^2 \mu_{S}(x) \, dx - J[\mu_{\sqrt{\lambda} S}, \mu_{\sqrt{\lambda} S}\boxplus_{\alpha} \mu_{\rm MP}]
\end{equation*}
To prove this limit, first, we show that this limit also holds for the free energy of a simpler model. Then, using the \textit{pseudo-Lipschitz} continuity of the free energy w.r.t. to a distance between two models which converges to $0$ as $N \to \infty$, we deduce that the same limit holds for the free energy of the original model.

\subsection{A simple model}
Suppose $\bsig^0 \in \bR^N$ is generated with i.i.d. elements from $\mu_S$. Fix $\bsig^0$ once for all. Construct the matrix $\tilde{\bSig} \in \bR^{N \times M}$from the vector $ \tilde{\bsig} \in \bR^N$. Construct the matrix $\tilde{\bS} \in \bR^{N \times M}$  as $\bU \tilde{\bSig} \bV^T$ where $\bU \in \bR^{N \times N}, \bV \in \bR^{M \times M}$ are independent and distributed according to the Haar measure. The distribution of the matrix $\tilde{\bS}$ is :
\begin{equation}
\begin{split}
        d P_{\tilde{S}}(\tilde{\bS}) &= d \mu_N(\bU) \,  d \mu_M(\bV) d p_{\tilde{S}}(\tilde{\bsig}) \\
        & \propto d \mu_N(\bU) \, d \mu_M(\bV) \, \prod_{i=1}^N \delta(\tilde{\sigma}_i - \sigma^0_i) \, d \tilde{\bsig}
\end{split}
\label{model 2}
\end{equation}
 
Matrix $\tilde{\bS}$ is observed through an AWGN channel as in \eqref{observation-matrix}, $\tilde{\bY} = \sqrt{\lambda} \tilde{\bS} + \tilde{\bZ}$. The partition function and the free energy can be defined in the same way as in \eqref{partition-function-def},\eqref{free-energy-def} denoted by $\tilde{Z}(\tilde{\bY})$, $\tilde{F}_N(\lambda)$ respectively.
\begin{proposition}
For $\mu_S$ with compact support, and any $\lambda > 0$, we have $\mu_S$-almost surely
\begin{equation*}
    \lim_{N \to \infty} \tilde{F}_N(\lambda) = \frac{\lambda}{2} \alpha \int x^2 \rho_{S}(x) \, dx - J[\mu_{\sqrt{\lambda} S}, \mu_{\sqrt{\lambda} S}\boxplus_{\alpha} \mu_{\rm MP}] 
\end{equation*}
\label{asymp-free-energy-2-proposition}
\end{proposition}
\vspace{-5mm}
\textit{Proof.} Appendix \ref{proof of prop1}.

\subsection{Pseudo-Lipschitz continuity of the free energy}
Consider two rotationally invariant matrix ensemble $P^{(1)}$, $P^{(2)}$, i.e. for $\bS \sim P^{(1)}(\bS)$, $\tilde{\bS} \sim P^{(2)}(\tilde{\bS})$ with SVDs $\bS = \bU \bSig \bV^T$, $\tilde{\bS} = \tilde{\bU} \tilde{\bSig} \tilde{\bV}^T$
\begin{equation*}
\begin{split}
        d P_N^{(1)} (\bS )  &\propto d \mu_N(\bU) \, d \mu_M(\bV) \, p^{(1)} (\bsig ) \, d \bsig \\
        d P_N^{(2)} ( \tilde{\bS} )  &\propto d \mu_N(\tilde{\bU}) \, \propto d \mu_M(\tilde{\bV})\, p^{(2)} (\tilde{\bsig} ) \, d \tilde{\bsig}
\end{split}
\end{equation*}
where  $p^{(1)}(\bsig)$, $p^{(2)}(\tilde{\bsig})$ are the joint probability density functions for the singular values, induced by the priors. Suppose each of these distributions to be the prior of an inference problem in model \eqref{observation-matrix}. The free energy can be defined similarly for each of the priors, which are denoted by $F_N^{(1)}(\lambda), F_N^{(2)}(\lambda)$ respectively. Then, we have

\begin{proposition}
For all $\lambda > 0$ and $N$ :
\begin{align}
    &\big| F_N^{(1)}(\lambda) - F^{(2)}_N(\lambda) \big|  \label{pseudo-lip-ub}\\
    &\leq \frac{\lambda}{2 N} \Big( \sqrt{ \bE_{\bsig} \big[ \| \bsig \|^2 \big]} + \sqrt{ \bE_{\tilde{\bsig}} \big[ \| \tilde{\bsig} \|^2 \big]} \Big) \sqrt{\bE_{\bsig, \tilde{\bsig}} \big[ \| \bsig - \tilde{\bsig} \|^2 \big] }   \notag
\end{align}
\label{pseudo-lip}
\end{proposition}
\vspace{-5mm}
Proof of this proposition is similar to the proof of Proposition B.1 in \cite{miolane2019fundamental}.

\subsection{The distance between two models}
Recall that
\begin{equation*}
\begin{split}
        d P_{S}(\bS)  \propto d \mu_N(\bU) \, d \mu_M(\bV) p_{S}(\bsig) \, d \bsig^S
\end{split}
\end{equation*}
where $p_{S}(\bsig)$ is the joint p.d.f. of singular values of $\bS$. Moreover, $d P_{\tilde{S}}(\tilde{\bS})$ is defined in \eqref{model 2} with  $p_{\tilde{S}}(\tilde{\bsig}) \equiv \prod_{i=1}^N \delta(\tilde{\sigma}_i - \sigma^0_i)$, where $\bsig^0$ is generated with i.i.d. elements from $\mu_S$
\begin{lemma}
Under assumptions \ref{assumptions on law}, \ref{bounded-mom}, for $\bsig \sim p_{S}(\bsig)$, $\tilde{\bsig} \sim p_{\tilde{S}}(\tilde{\bsig})$ , we have:
\begin{equation}
    \lim_{N \to \infty} \frac{1}{N} \bE_{\bsig, \tilde{\bsig}} \big[ \| \bsig - \tilde{\bsig} \|^2 \big] =0
    \label{W2-dis-eq}
\end{equation}
\label{W2-dis}
\end{lemma}
\vspace{-5mm}
\textit{Proof.} Appendix \ref{proof of lemma 2}.

\subsection{Proof of theorem \ref{main-th}}
By proposition \ref{pseudo-lip},  the distance between the free energies $F_N(\lambda)$ (defined in \eqref{free-energy-def}) and $\tilde{F}_N(\lambda)$ is upper bounded by rhs in \eqref{pseudo-lip-ub}.
The term $\frac{1}{N} \| \bsig \|^2 = \frac{1}{N} \sum \sigma_i^2 $ is the second moment of the empirical spectral distribution of $\bS$, which is almost surely bounded by assumption \ref{bounded-mom}. So,  $\frac{1}{N} \bE_{\bsig} \big[ \| \bsig \|^2 \big]$ is bounded.
 Moreover, $\frac{1}{N} \bE \big[ \| \tilde{\bsig} \|^2 \big] = \frac{1}{N}  \sum {\sigma^0}_i^2$ which is bounded by $C_2^2$. By proposition \ref{W2-dis}, $\lim_{N \to \infty} \frac{1}{N} \bE_{\bsig, \tilde{\bsig}} \big[ \| \bsig - \tilde{\bsig} \|^2 \big] = 0$. Therefore $    \lim_{N \to \infty} |F_N(\lambda) - \tilde{F}_N(\lambda) | = 0$ and Proposition \ref{asymp-free-energy-2-proposition} gives the result. $\hfill \square$

\section{Detailed proof of Theorem 1}\label{proof-thm1-app}
\subsection{Proof of proposition \ref{asymp-free-energy-2-proposition}}\label{proof of prop1}
We start from the partition function,\vspace{-2mm}
\begin{align}
        &\tilde{Z}(\tilde{\bY}) =  \int d \bX e^{N \Tr \big[ \sqrt{\lambda}\bX^T \tilde{\bY} - \frac{\lambda}{2} \bX^T \bX \big] } P_{\tilde{S}}(\bX) \notag \\
        &= \iiint d \bsig \,d \mu_N(\bU) \, d \mu_M(\bV) \, \prod_{i=1}^N \delta(\sigma_i - \sigma^0_i) \notag\\
        & \hspace{30pt}\times  e^{N \Tr [ \sqrt{\lambda} \bV \bSig^T \bU^T \tilde{\bY} - \frac{\lambda}{2}\bSig \bSig^T ]} \notag\\
        &= e^{ - \frac{N}{2} \lambda \Tr {\bSig^0}^T {\bSig^0} } \iint d \mu_N(\bU) \,  d \mu_M(\bV) \, e^{N \Tr [ \sqrt{\lambda} {\bSig^0}^T \bU  \tilde{\bY} \bV^T ]}\notag \\
        &=  e^{ - \frac{N}{2} \lambda \Tr  {\bSig^0}^T {\bSig^0}  } I_N \big( \sqrt{\lambda} \bSig^0, \tilde{\bY} \big)
\label{Z2-comp1}
\end{align}\vspace{-6mm}

\noindent Note that, we change variables $\bU \to \bU^T, \bV \to \bV^T$ in third line to match the definition of the spherical integral.

Recall that $\tilde{Y} = \sqrt{\lambda} \bU \bSig^0 \bV^T + \tilde{\bZ}$, so the free energy can be written as:\vspace{-2mm}
\begin{align}
        \tilde{F}_N(\lambda) &= \bE_{\tilde{\bY}} \Big[ \frac{\lambda}{2 M} \Tr {\bSig^0}^T {\bSig^0}  - J_N\big( \sqrt{\lambda} \bSig^0, \tilde{\bY} \big) \Big] \label{F2-com} \\
        &\hspace{-25pt}=  \frac{\lambda}{2 M} \sum_{i=1}^N {\sigma^0_i}^2 - \bE_{\bU,\bV,\tilde{\bZ}} \Big[ J_N\big( \sqrt{\lambda} \bSig^0, \sqrt{\lambda} \bU \bSig^0 \bV^T + \tilde{\bZ} \big) \Big] \notag
\end{align}
\vspace{-4mm}

\noindent By rotational invariance of $\tilde{\bZ}$, the second term equals \vspace{-2mm}
\begin{equation*}
\begin{split}
        \bE_{\bU,\bV,\tilde{\bZ}} \Big[ J_N\big( \sqrt{\lambda} \bSig^0, \sqrt{\lambda} \bU \bSig^0 \bV^T + \bU \tilde{\bZ} \bV^T \big) \Big]
\end{split}
\end{equation*}
\vspace{-5mm}

\noindent and then, both matrices $\bU, \bV$ can be absorbed into the integration in $J_N$. So, the free energy equals:\vspace{-2mm}
\begin{equation*}
     \tilde{F}_N(\lambda) = \frac{\lambda}{2 M} \sum_{i=1}^N {\sigma^0_i}^2 - \bE_{\tilde{\bZ}} \Big[ J_N\big( \sqrt{\lambda} \bSig^0, \sqrt{\lambda} \bSig^0  +  \tilde{\bZ}\big) \Big]
\end{equation*}
\vspace{-3mm}

By the strong law of large numbers, the first term in \eqref{F2-com} converges to $\frac{\lambda}{2} \alpha \int x^2 \mu_S(x) \, d x$ almost surely, and proposition \ref{asymp-free-energy-2-proposition} follows from the following lemma.
\begin{lemma}
For any $\lambda \in \bR_+$, the sequence $\bE_{\tilde{\bZ}} \Big[ J_N\big( \sqrt{\lambda} \bSig^0, \sqrt{\lambda} \bSig^0  +  \tilde{\bZ}\big) \Big]$ converges to $J[\mu_{\sqrt{\lambda} S}, \mu_{\sqrt{\lambda} S}\boxplus_{\alpha} \mu_{\rm MP}]$ as $N \to \infty$, $\mu_S$-almost surely.
\label{converg-Expec}
\end{lemma}
\vspace{-5mm}
\textit{Proof sketch.} We show that the assumptions of Theorem 1.1 in \cite{guionnet2021large} holds a.s. for the sequence $\sqrt{\lambda} \bSig^0, \sqrt{\lambda} \bSig^0  +  \tilde{\bZ}\big)$, so $J_N$ converges to $J$ a.s. . To show that the limit also holds under the expectation $\bE_{\tilde{\bZ}}$, we use the fact the $J_N$'s are bounded (see lemma 14 in \cite{guionnet2005fourier} for the symmetric case) by the product of top singular values of $\sqrt{\lambda} \bSig^0$ and $\sqrt{\lambda} \bSig^0  +  \tilde{\bZ}\big)$, by triangle inequality \vspace{-2mm}
\begin{equation*}
    \Big| J_N \big( \sqrt{\lambda} \bSig^0,\sqrt{\lambda} \bSig^0  +  \tilde{\bZ}\big) \Big| \leq \sqrt{\lambda} C_2 \big( \sqrt{\lambda} C_2 + \sigma_1^{\tilde{Z}} \big)
\end{equation*}
\vspace{-4mm}

\noindent By \cite{davidson2001local}, $\bE[\sigma_1^{\tilde{Z}}]$ converges to $1+1/\sqrt{\alpha}$. So, the convergence in expectation of $J_N$ follows from dominated convergence theorem.

\subsection{Proof of lemma \ref{W2-dis}}\label{proof of lemma 2}
First, note that by rotational invariance, $p_{S}(\bsig)$ is invariant under permutations, so without loss of generality, we can assume $\bsig$ is in non-decreasing order.

Since $p_{\tilde{S}}(\tilde{\bsig})$ is a delta distribution, we can easily write
\begin{equation}
    \bE_{\bsig, \tilde{\bsig}} \big[ \| \bsig - \tilde{\bsig} \|^2 \big] =  \bE_{\bsig} \big[ \| \bsig - \bsig^0 \|^2 \big]
    \label{W-2-1}
\end{equation}

For a vector $\bsig$, denote the empirical distribution of its components by $\hat{\mu}_{\bsig}$. The Wasserstein-2 distance between two empirical distributions, $\hat{\mu}_{\bsig}, \hat{\mu}_{\bsig^0}$ is defined as 
\begin{equation*}
\begin{split}
    W_2(\hat{\mu}_{\bsig}, \hat{\mu}_{\bsig^0}) = \sqrt{\inf_{\gamma \in \Gamma(\hat{\mu}_{\bsig}, \hat{\mu}_{\bsig^0})} \bE_{\gamma(x, y)}\big[ ( x  - y)^2 \big]}
\end{split}
\end{equation*}
with $\Gamma(\hat{\mu}_{\bsig}, \hat{\mu}_{\bsig^0})$ is the set of couplings of $\hat{\mu}_{\bsig}, \hat{\mu}_{\bsig^0}$.By the Example in page 5 in \cite{villani2021topics} ,the Wasserstein-2 can be written as
\begin{equation}
\begin{split}
    W_2(\hat{\mu}_{\bsig}, \hat{\mu}_{\bsig^0}) = \sqrt{\min_{\pi \in \mathcal{S}_N}  \frac{1}{N} \| \bsig  - \bsig^0_{\pi} \|^2 }
\end{split}
\label{Wasser-equiv}
\end{equation}
where $\bsig^0_{\pi}$ is the permuted version of $\bsig^0$, and $\mathcal{S}_N$ is the set of all $N$-permutations. So, for a given $\bsig$ and $\bsig^0$ (which have a non-decreasing order), we have (considering the identity permutation)
\begin{equation}
    \| \bsig - \bsig^0 \|^2 \geq N W_2(\hat{\mu}_{\bsig}, \hat{\mu}_{\bsig^0})^2
        \label{W-2-2}
\end{equation}

On the other hand, for any permutation of $\bsig^0$ (in particular, the one which achieves the minimum in \eqref{Wasser-equiv}), we have
\begin{equation*}
    \begin{split}
        \| \bsig  - \bsig^0_{\pi} \|^2 &= \| \bsig \|^2 + \| \bsig^0_{\pi} \|^2 - 2 \bsig^T \bsig^0_{\pi} \\
        & \geq \| \bsig \|^2 + \| \bsig^0_{\pi} \|^2 - 2 \bsig^T \bsig^0 = \| \bsig  - \bsig^0 \|^2
    \end{split}
\end{equation*}
where we used rearrangement inequality \cite{hardy1952inequalities} to get the inequality in the second line. So,
\begin{equation}
    \| \bsig - \bsig^0 \|^2 \leq N W_2(\hat{\mu}_{\bsig}, \hat{\mu}_{\bsig^0})^2
        \label{W-2-3}
\end{equation}

From \eqref{W-2-1}, \eqref{W-2-2},\eqref{W-2-3}, we have
\begin{equation}
  \bE_{\bsig, \tilde{\bsig}} \big[ \| \bsig - \tilde{\bsig} \|^2 \big] = \bE_{\bsig} \big[N W_2(\hat{\mu}_{\bsig}, \hat{\mu}_{\bsig^0})^2 \big]
    \label{W-2-5}
\end{equation}
Lemma \ref{E-W-0-l} concludes the proof.



\begin{lemma}\label{E-W-0-l}
Suppose $\bsig \in \bR^N$ is distributed according to $p_{S}(\bsig)$, and $\bsig^0$ is generated with i.i.d. elements from $\mu_S$. Let $\hat{\mu}_{\bsig}, \hat{\mu}_{\bsig^0}$ be their empirical distribution. We have:
\begin{equation*}
    \lim_{N \to \infty} \bE_{\bsig} \big[ W_2(\hat{\mu}_{\bsig}, \hat{\mu}_{\bsig^0})^2 \big] = 0
\end{equation*}
\end{lemma}
\begin{proof}
By triangle inequality, we have:
\begin{equation}
    W_2(\hat{\mu}_{\bsig}, \hat{\mu}_{\bsig^0}) \leq W_2(\hat{\mu}_{\bsig}, \mu_S) + W_2(\hat{\mu}_{\bsig^0}, \mu_S)
\end{equation}
From the weak convergence and convergence of second moment, assumptions \ref{assumptions on law} and \ref{bounded-mom} imply that the second moment of the empirical spectral distribution converges almost surely to the one of $\mu_S$. 
Thus, by  \cite{villani2021topics}(Theorem 7.12),  the empirical singular value distribution in the Wasserstein-2 metric to  $\mu_S$. Hence, the first term approaches $0$ as $N \to \infty$ almost surely.

By law of large numbers and since the support of $\mu_S$ is bounded, the second term also converges $0$ as $N \to \infty$. Therefore, we have $W_2(\hat{\mu}_{\bsig}, \hat{\mu}_{\bsig^0}) \to 0$ almost surely. Consequently, we have that $W_2(\hat{\mu}_{\bsig}, \hat{\mu}_{\bsig^0})^2 \to 0$ almost surely.

One can see that:
\begin{equation}
    \begin{split}
        W_2(\hat{\mu}_{\bsig}, \hat{\mu}_{\bsig^0})^2 & \leq \frac{2}{N} \sum \sigma_i^2 + \frac{2}{N} \sum {\sigma^0_i}^2 \\
        &\leq 2  m^{(2)}_{\hat{\mu}_{\bsig}} + 2 C_2^2
    \end{split}
\end{equation}
with $m^{(2)}_{\hat{\mu}_{\bsig}}$ the second moment of $\hat{\mu}_{\bsig}$ which is almost surely bounded by assumption. Therefore, the result follows by using dominated convergence theorem.
\end{proof}




\end{document}